\documentclass[10pt, journal]{IEEEtran}
\IEEEoverridecommandlockouts
\usepackage{amssymb,amsfonts,bm}
\usepackage{amsthm}
\usepackage{balance}
\usepackage{cite}
\usepackage{amsmath,amssymb,amsfonts}
\usepackage{algorithmic}
\usepackage[shortlabels]{enumitem}
\usepackage{textcomp}
\usepackage{xcolor}
\usepackage{subfig}
\usepackage{graphicx}
\usepackage{tikz}
\usepackage{comment}
\usepackage{cite}
\newtheorem{theorem}{Theorem}[]

\newtheorem{lemma}[]{Lemma}
\newtheorem{claim}[]{Claim}
\newtheorem{definition}{Definition}[]
\newtheorem{proposition}{Proposition}

\usepackage{lipsum}

\newcommand\blfootnote[1]{%
  \begingroup
  \renewcommand\thefootnote{}\footnote{#1}%
  \addtocounter{footnote}{-1}%
  \endgroup
}

\newcounter{relctr} %% <- counter for relations
\everydisplay\expandafter{\the\everydisplay\setcounter{relctr}{0}} %% <- reset every eq
\newcommand\labelrel[2]{%
  \begingroup
    \refstepcounter{relctr}%
    \stackrel{\textnormal{(\alph{relctr})}}{\mathstrut{#1}}%
    \originallabel{#2}%
  \endgroup
}
\AtBeginDocument{\let\originallabel\label} %% <- store original definition

\def\BibTeX{{\rm B\kern-.05em{\sc i\kern-.025em b}\kern-.08em
    T\kern-.1667em\lower.7ex\hbox{E}\kern-.125emX}}
\begin{document}

\title{Finite Time Bounds for Stochastic Bounded Confidence Dynamics\\
%{%\footnotesize \textsuperscript{*}Note: Sub-titles are not captured in Xplore and should not be used}
}

\author{%
  \IEEEauthorblockN{Sushmitha Shree S, Avhishek Chatterjee, Krishna Jagannathan}\\
  \IEEEauthorblockA{Department of Electrical Engineering\\  Indian Institute of Technology Madras, Chennai 600036, India\\
                    \{ee18d702\}@smail.iitm.ac.in, 
                    \{avhishek, krishnaj\}@ee.iitm.ac.in 
                    }
                    }

\maketitle
\blfootnote{A preliminary version of this paper appeared in the proceedings of COMmunication Systems \& NETworkS (COMSNETS) 2022 \cite{Shree2022}.}
\begin{abstract}
In this era of fast and large-scale opinion formation, a mathematical understanding of opinion evolution, a.k.a. opinion dynamics, acquires importance. Linear graph-based dynamics and bounded confidence dynamics are the two popular models for opinion dynamics in social networks. Stochastic bounded confidence (SBC) opinion dynamics was proposed as a general framework that incorporates both these dynamics as special cases and also captures the inherent stochasticity and noise (errors) in real-life social exchanges. Although SBC dynamics is quite general and realistic, its analysis is  more challenging. This is because SBC dynamics is nonlinear and stochastic, and belongs to the class of Markov processes that have asymptotically zero drift and unbounded jumps. The asymptotic behavior of SBC dynamics was characterized in prior works. However, they do not shed light on its finite-time behavior, which is often of interest in practice. We take a stride in this direction by analyzing the finite-time behavior of a two-agent system and a bistar graph, which are crucial to the understanding of general multi-agent dynamics. In particular, we show that the opinion difference between the two agents is well-concentrated around zero under the conditions that lead to asymptotic stability of the SBC dynamics.
\end{abstract}

\begin{IEEEkeywords}
Opinion dynamics; Markov process; Concentration inequality.

\end{IEEEkeywords}

\section{Introduction}
\label{sec:intro}
Public opinion is the driving force of a society. The advent of social media platforms has revolutionized the speed and scale of opinion formation, resulting in significant effects on societies. Hence, modeling opinion formation, popularly known as opinion dynamics, and analyzing its behavior is a very important problem.

The study of opinion dynamics has a long history \cite{Noorazar2020}. In the mathematical and computational study of opinion dynamics, individuals or social entities, a.k.a. agents, are modeled to have real-valued opinions regarding a topic. A positive (negative) opinion represents a favorable (unfavorable) view of the topic, and its magnitude represents the agent's conviction. Opinion dynamics models are discrete-time dynamical systems where the opinions of the agents at the next time slot are updated according to a specified function of the current opinions.

\subsection{Background and Motivation}
Broadly, there have been two popular models of opinion dynamics: linear graph-based dynamics and bounded confidence dynamics.  In the first model \cite{French1956, ABELSON1967, DeGroot1974, Friedkin1999, Fagnani2007, Acemoglu2013, Salehi2010, Yildiz2013}, opinion updates occur according to a linear combination of opinions of neighbors on a social graph.  In the original bounded confidence dynamics \cite{Axelrod1997,Deffuant2000,Hegselmann2002}, an agent updates its opinion using the average of the opinion of all agents (including itself) whose opinions are within a specified distance from its own opinion. Thus, in short, linear dynamics considers social graph-based opinion exchanges, whereas bounded confidence dynamics considers opinion-dependent opinion exchanges. 

Although the study of these two models and their variations \cite{French1956, ABELSON1967, Axelrod1997, Friedkin1999, Hegselmann2002, Fagnani2007, Acemoglu2013, Salehi2010, Yildiz2013, Noorazar2020} span most of the literature, they capture only partially, the nuances of real-world social interactions. An important characteristic of interactions in the real world is that they are often stochastic.  Since no agent can know the true opinion of the other agent and can, at best, estimate it based on expressed views.  Hence, whether an agent accepts or rejects another agent's opinion is a random phenomenon.  Naturally, the probability of accepting another agent's opinion is lower when the difference between their (true) opinions is larger.  Moreover, while updating its own opinion based on other accepted opinions, an agent incorporates {estimates} of others' opinions tempered by its own inherent beliefs.  This leads to unavoidable {\em noise} in the opinion updates. 

Stochastic bounded confidence (SBC) opinion dynamics proposed in \cite{Baccelli2015, Baccelli-Infocom, Baccelli2017} is a framework that addresses these real-world issues while merging the graph-based exchanges in linear dynamics with a stochastic generalization of the opinion-dependent exchanges in bounded confidence dynamics.  Unlike prior opinion dynamics models, where opinions eventually converge, the SBC dynamics captures real-life scenarios where opinions in a society, depending on the scenario, may stay close or diverge to opposite extremes.

Linear dynamics and bounded confidence dynamics have been analyzed in detail in the literature.  SBC dynamics is a stochastic generalization of both these dynamics on a graph.  Though the nonlinear and stochastic nature of SBC dynamics makes this model much more realistic, its analysis becomes significantly more challenging.  In \cite{Baccelli2015, Baccelli-Infocom, Baccelli2017}, specific conditions involving the social graph and the nature of the stochastic opinion-dependent exchanges were provided for limiting opinion differences to be finite.  In multiple settings, tight converse results were also provided.  Although these results display significant initial progress, they do not disclose anything about the evolution of opinions over a finite time window, which is often of interest in practice. 

\subsection{Our Contributions}
In this paper, we take the first stride towards a non-asymptotic characterization of the opinion differences under SBC dynamics. Here, we obtain high probability bounds on opinion differences at a finite time for stable SBC dynamics. The primary focus of this work is on characterizing the evolution of the opinion difference for \emph{two-agent} SBC dynamics at a finite time. We derive high probability bounds for the opinion difference under sub-Gaussian noise using a Chernoff bound. In particular, we demonstrate that the opinion difference between the agents is well-concentrated around zero under the same technical conditions that imply the asymptotic stability of the SBC dynamics.  

Next, motivated by a two-party democratic polity, we study SBC dynamics on a bistar social network\footnote{A bistar graph is a reasonable model for a polity with two leaders/parties having large followings.}. We obtain bounds for a class of asymptotically stable SBC dynamics on bistar social graphs using the insights gathered from the two-agent dynamics.  

The reason for starting with the two-agent case is to study the issues of nonlinearity and stochasticity in isolation from the graph structure. Also, the two-agent SBC dynamics is a necessary building block for understanding the dynamics on a general graph. The insights obtained from the two-agent dynamics help us to develop bounds for SBC dynamics on bistar social graphs. Towards the end, we support our theoretical analysis using numerical results.

% We analyzed the average opinion of agents over a time window and provide high probability bounds on difference of average opinions. We also present simulation results supplementing our analytical bounds.

\subsection{Organization}

In the following section (Sec.~\ref{sec:SBC}), we briefly discuss SBC dynamics. The two-agent stable SBC dynamics is analyzed in Sec.~\ref{sec:mainResult}. In this section, we present a high probability bound on the opinion difference at a finite time when the noise (errors) in opinion exchange has a sub-Gaussian distribution. In Sec.~\ref{sec:proofOutline}, we outline the proof of the main result of two-agent dynamics, starting with the relatively simpler case of bounded noise in Sec.~\ref{sec:boundedN}, followed by its extension to sub-Gaussian noise in Sec.~\ref{sec:subG}. We then present the high probability bounds for opinion differences in stable SBC dynamics on a bistar graph (Sec.~\ref{section:network}), followed by the numerical results in Sec.~\ref{simulations}. Detailed proofs are in Appendix.

\section{Stochastic bounded confidence dynamics}
\label{sec:SBC}

Stochastic bounded confidence (SBC) opinion dynamics \cite{Baccelli2017} captures the effect of the social graph as well as that of the closeness of opinions on opinion evolution. Furthermore, it models the impact of inherent stochasticity in human interactions and the unavoidable noise and errors in opinion exchanges. It is a general framework for opinion dynamics that captures the well-known linear dynamics and the bounded confidence dynamics as special cases. Below, we briefly describe a simplified version of the dynamics that is sufficient for the purpose of this work. Please see \cite{Baccelli2017} for a more general description. 

In SBC dynamics, there is an underlying undirected social network $\mathcal{G}=(\mathcal{V}, \mathcal{E})$ of $n$ agents, which captures the possible interactions. Agent $u$ can interact with $v$ only if it has an edge $(u,v)$ with $v$. For each edge $(u,v) \in \mathcal{E}$, there are two influence functions $G_{u,v}, G_{v,u}:[0,\infty) \to [0,1]$, which capture the probability of influence of $u$ on $v$ and that of $v$ on $u$, respectively, as functions of their opinion difference.

Opinions are real-valued and evolve as discrete-time stochastic dynamics. Opinions at time $t$ are denoted by $\{X_u(t): u \in \mathcal{V}\}$. Any two agents $u$ and $v$ with $(u,v) \in \mathcal{E}$ interact at time $t$ with a non-zero probability, and at any time, an agent interacts with at most one other agent. If $u$ and $v$ interact at time $t$, then agent $u$ influences $v$ and vice versa with probabilities $G_{u,v}(|X_u(t)-X_v(t)|)$ and $G_{v,u}(|X_u(t)-X_v(t)|)$, respectively. 

If $u$ influences $v$ at time $t$, then $v$ updates its opinion as
\begin{align*} 
%X_u(t+1) &= \frac{X_u(t)+X_v(t)}{2} + n_u(t), \nonumber \\
X_v(t+1) &= \frac{X_u(t)+X_v(t)}{2} + n_v(t), \nonumber
\end{align*}
and if $v$ is not influenced by any agent, then its opinion evolves as
\begin{align*} 
X_v(t+1) &= X_v(t) + n_v(t). \nonumber
\end{align*}
Here, for each agent $u$, $n_u(t)$ is an i.i.d. zero mean process. This captures the errors and noise in the interactions, which stem from miscommunications and misinterpretations. This also captures the innate evolution of the opinion of an agent due to his/her own thoughts and emotions. 

Note that if one chooses $G_{u,v}(x)=G_{v,u}(x)=1$ for all $x$ and $(u,v) \in \mathcal{E}$ and $n_u(t)=0$, we get back the well-known linear dynamics on a social network. On the other hand, if we choose $\mathcal{G}$ to be a clique and for any $u,v$, choose $G_{u,v}(x)=G_{v,u}(x)=1$ for $x\le d$  and $G_{u,v}(x)=G_{v,u}(x)=0$ for all $x>d$, we get back the well-known pairwise bounded confidence opinion dynamics. Thus, these two popular classes of dynamics are special cases of SBC dynamics.

Due to nonlinearity and stochasticity,  analyzing SBC is significantly more challenging. Further, due to the presence of noise (or estimation error) in the SBC dynamics, a consensus cannot be reached, not even in an almost sure or a high probability sense. This, in a way, reflects the real social dynamics, where there may not be a consensus. In such a scenario, just like in real democratic societies, we can, at best, hope for the differences of opinions to remain finite. To capture this scenario, the notion of stability of opinion dynamics was introduced in \cite{Baccelli2017}, and conditions for the stability of SBC dynamics were established. 

Mathematically, the stability of SBC dynamics is defined as the opinion differences between agents asymptotically reaching a proper stationary distribution. On the other, the dynamics is said to be not stable if the opinion differences do not reach a stationary distribution. These capture the cases when the opinions of two groups stay close and diverge away, respectively. The stability results (and their converses) in \cite{Baccelli2017} provide simple conditions in terms of $\mathcal{G}$ and $\{G_{u,v}\}$ that lead to stability (and otherwise). 

The stability results are essential in understanding the dynamics. Still, due to their asymptotic nature, they do not shed much light on opinion differences at a finite time, which are often of practical interest. However, standard concentration inequalities cannot be applied to obtain bounds for opinion differences in stable SBC dynamics due to their jumpy, nonlinear behavior.

\section{Bounds for Two-agent Dynamics}
\label{sec:mainResult}
For developing insights into the general dynamics, we first consider the following simple two-agent symmetric dynamics. Agents $1$ and $2$ interact at all time instants. Thus, at time $t$, if their opinions are $X_1(t)$ and $X_2(t)$, then they are mutually influenced by each other with probability $G(|X_1(t)-X_2(t)|)$. Upon influence, their opinions are updated as, for $i=1,2$,
\[
X_i(t+1) = \frac{X_1(t)+X_2(t)}{2} + n_i(t)\]
and when they are not influenced, the opinions evolve as, for $i=1,2$,
\begin{align*} 
X_i(t+1) &= X_i(t) + n_i(t). \nonumber
\end{align*}

Their opinion difference $Y(t):=X_1(t)-X_2(t)$ is a discrete-time stochastic process and its evolution can be written as: given $Y(t)=y$, $Y(t+1)=\tilde{n}(t)$ with probability $G(|y|)$ and with probability $1-G(|y|)$,
\[Y(t+1) = Y(t) + \tilde{n}(t),\]
where $\tilde{n}(t)=n_1(t)-n_2(t)$ is the difference of two independent zero mean i.i.d. noise processes. We assume that $\tilde{n}(t)$ has a symmetric distribution about its mean.

It was shown in \cite{Baccelli2017} that the opinion difference $Y(t)$ reaches a stationary distribution, i.e., the SBC dynamics is stable, if for some $\delta>0$, $G(x) \gtrsim \frac{1}{x^{2-\delta}}$, where the notation $m(x)\gtrsim g(x)$ means $\liminf_{x\to\infty} \frac{m(x)}{g(x)}>0$. It was also shown that the dynamics is not stable if for some $\delta>0$, $\frac{1}{x^{2+\delta}} \gtrsim G(x)$. The influence function $G(.)$ is monotonically non-increasing. We define $G_0:=G(0)$.

It is not hard to see that if $Y(0)=0$, for most two-agent opinion dynamics, including the SBC dynamics, $Y(t)$ can be written as a function $f_t(\{\tilde{n}(\tau), U_{\tau}: 1 \le \tau \le t\})$, where $U_i$ are i.i.d. uniform $[0,1]$ random variables. Moreover, if $Y(0)=0$, $\mathbb{E}[Y(t)]=\mathbb{E}[f_t(\{\tilde{n}(\tau), U_{\tau}: 1 \le \tau \le t\})]=0$. Hence, it may seem that an application of McDiarmid like inequalities should give a tight high probability bound on $|Y(t)|$. However, for SBC dynamics, $f_t$ has unbounded discontinuities, and hence, they require different treatments. For the same reason, the well-known Markov concentration results \cite{Kavitha} cannot be applied either.

In this paper, our first important result is a bound on $|Y(t)|$ in stable two-agent dynamics at a finite $t$. Later, we build on these intuitions and proof techniques to obtain bounds for dynamics on the bistar graph. We establish the bound for the class of sub-Gaussian i.i.d. noise processes.

\begin{definition}[{Sub-Gaussian Random Variable \cite[Sec. 2.3]{lugosi}}]

A random variable $X$ with $\mathbb{E}[X]=0$ is sub-Gaussian with variance parameter $\sigma^2$, denoted by $X \in \mathcal{SG}(\sigma^2)$, if $\; \forall \lambda \in \mathbb{R}$,
\begin{align*}
    \mathbb{E}[\exp{(\lambda X)}] \leq \exp{\Big(\frac{\lambda^2 \sigma^2}{2}\Big)}.
\end{align*} 
% which implies that for every $m>0$,
% \begin{align*}
%     \mathbb{P}(X>m)&\leq \exp{\Big(-\frac{m^2}{2\sigma^2}\Big)}\\\mathbb{P}(X<-m)&\leq \exp{\Big(-\frac{m^2}{2\sigma^2}\Big)}
% \end{align*}
\end{definition}

The main result in the case of two-agent dynamics is a high probability bound on $|Y(t)|$ at any finite time $t$.

\begin{theorem}\label{theorem_sg}
Consider a two-agent stable dynamics with $G(x) \gtrsim \frac{1}{x^{2-\delta}}$ for some $\delta>0$, and i.i.d.  $\tilde{n}(t)\in\mathcal{SG}(\sigma^2)$ for some finite $\sigma$. Let $k=c~t^{\frac{1}{2}-\beta}$ for $c, \beta>0$. Then, with $d_{\tau}=D~\tau^{\frac{1}{2}+\beta'}$ for some $\beta'>0$, $D>0$, and $c'>0$, and $h(t)<t$,
\begin{align*}
    \mathbb{P}_0(|Y(t)| \geq k)&\leq 2(t-h(t))\exp{(-c'h(t)^{ 2\beta'})}+2\Big[\frac{t-h(t)}{1-G(d_t)}\nonumber\\ &+\exp{\Big(G(d_t)h(t)\Big)}\Big] \exp{\Big(- \frac{\sqrt{2 G(d_t)}}{\sigma}k\Big)}
\end{align*}
for all $t\ge 0$. Here, $\mathbb{P}_0(\cdot)$ corresponds to the conditional probability given that the initial difference $Y(0)=0$.
% Consider a two-agent stable dynamics with $G(x) \gtrsim \frac{1}{|x|^{2-\delta}}$ for some $\delta>0$, and sub-Gaussian noise model with the assumed characteristics. Then, with $d_{\tau}=D\tau^{\frac{1}{2}+\beta'}$ for some $\beta'>0$, there exists $\beta<\frac{\delta}{4}$, $\zeta<1-\frac{\delta}{2}$ for which the tail probability of opinion difference decays exponentially with time $t$. That is, for $c>0$, $\gamma(\lambda)=\mathcal{O}(1), \lambda=\mathcal{O}(t^{\bar{k}(\frac{\delta}{2}-1)}), \alpha=\mathcal{O}(t^{\bar{k}(\delta-2)})$ where $\bar{k}=\frac{1}{2}+\beta'$,
% \begin{multline*}
%     \mathbb{P}_0(|Y(t)| \geq k) \leq 2(t-h(t))\exp{(-ch(t)^{2\beta'})}\\+2\Big[\exp{\Big(\frac{\lambda^2\sigma^2h(t)}{2}}\Big)+\frac{\gamma(\lambda)}{\alpha}\Big]
%     \exp{(-\lambda k)}
% \end{multline*}
% where, $\mathbb{P}_0$ is the probability measure conditioned on $Y(0)=0$.
%  {\color{red}Replace with the better theorem for general $d_{\tau}$.}%.......the upper bound decays for $\beta<\frac{\delta}{2}$ and $\zeta<1-\frac{\delta}{2}$.
\end{theorem}

The bound in Theorem~\ref{theorem_sg} can be simplified when $t$ is sufficiently large.

\begin{proposition}
\label{prop:mainresult}

Consider a two-agent stochastic bounded confidence dynamics with $G(x) \gtrsim \frac{1}{x^{2-\delta}}$ for some $\delta>0$ and i.i.d.  $\tilde{n}(t)\in\mathcal{SG}(\sigma^2)$ for some finite $\sigma$. Let $k=c~t^{\frac{1}{2}-\beta}$ for $c, \beta>0$ and $d_t=D~t^{\frac{1}{2}+\beta'}$ for some $\beta'>0$ and $D>0$. Then, for all $t> 0$ such that $\sqrt{t^{1-2\beta-2\epsilon}G(d_t)}\ge \frac{\theta\sigma}{\sqrt{2}c}$, \[\mathbb{P}_0(|Y(t)| \ge k) \le c_1 t \exp{(-\theta t^\epsilon)}\]
%{\color{red} Tighten the calculation and get an upper bound on $c_1$ in terms of $\theta$. }
for $\epsilon\le\frac{\delta}{6}-\frac{2\beta}{3}$ and positive (independent of $t$) constants $\theta$ and $c_1\le \frac{2(3-2G_0)}{1-G_0}$.
%{\color{red}$\le . . .$}. 

Furthermore, for $t>\Big(\frac{1}{\theta~\epsilon}\Big)^\frac{1}{\epsilon}$ and $c_2>0 $, %{\color{red}AC: please check the calculation once. For larger $\theta$, the bound is applicable to smaller $t$ and the decay is faster. This is counter intuitive.}
\[\mathbb{P}_0(\lvert Y(t) \rvert \geq k)\leq c_1 \exp{(-c_2 t^\epsilon)}.\]

%{\color{red} You may try to get a statement like the following. "Furthermore, for $t>\Big(\frac{1}{\theta~\epsilon}\Big)^\frac{1}{\epsilon}$ ....$\mathbb{P}_0(\lvert Y(t) \rvert \geq k)\leq c_1 \exp{(-0.9 \theta t^\epsilon)}$. Do similar tightening for other Propositions.}
\end{proposition}

Proposition~\ref{prop:mainresult} follows from Theorem~\ref{theorem_sg} for  $h(t)=t^\zeta$ with $\zeta\le1-\frac{\delta}{2}$ and $\epsilon\le\frac{\delta}{6}-\frac{2\beta}{3}$. 

% \begin{proposition}
% \label{cor:prop}
% {\color{red} Does not seem to bring out anything more than the previous result. Can be skipped.}
% For a two-agent stochastic bounded confidence dynamics with $G(x) \gtrsim \frac{1}{x^{1-\delta}}$, %{\color{red}Put the cleaned up version without $\beta'$ (choose a good one or optimize over $\beta'$ to find the best one, as I had told you, and get a lower bound on $c_3$.}
% \[\mathbb{P}_0(|Y(t)| \ge k) \leq c_1 t \exp{(-\theta t^\epsilon)}\]
% for $\epsilon\le\frac{\delta}{2}-\beta$ and all $t$ such that $\sqrt{t^{1-2\beta-2\epsilon}G(Dt)}\ge \frac{\theta D}{\sqrt{2}c}$. Also, for $t>\Big(\frac{1}{\theta~\epsilon}\Big)^\frac{1}{\epsilon}$ and $c_2>0$,
% \[\mathbb{P}_0(\lvert Y(t) \rvert \geq k)\leq c_1 \exp{(-c_2 t^\epsilon)}.\]

% \end{proposition}

First and obvious implication of Proposition~\ref{prop:mainresult} is that the concentration of the opinion difference around $0$ in a stable SBC dynamics is much tighter than that of sum of i.i.d. noise given by $S(t)=\sum_{\tau=1}^t n^{(b)}(\tau)$. Since it is well known that for any $\epsilon>0$ and $c>0$ \cite{lugosi},
\[\mathbb{P}(|S(t)| \ge c~t^{\frac{1}{2}-\epsilon})=1-o(1).\]

The behavior of sum of i.i.d noise is a benchmark of interest to us since the SBC dynamics would behave like that if $G(x)=0$ for all $x>0$. Thus, Proposition~\ref{prop:mainresult} shows that $G(x) \gtrsim \frac{1}{x^{2-\delta}}$, for some $\delta>0$, ensures asymptotic stability as well as closeness of opinions at finite time.

A direct corollary of Proposition~\ref{prop:mainresult} is that, for some $a, b>0$,
\[\mathbb{P}_0(\bigcup_{\tau=t}^\infty \{|Y(\tau)| \ge c~\tau^{\frac{1}{2}-\beta}\}) \le a \exp(-b t^{\epsilon})\]
for $\epsilon\le {\frac{\delta}{6}-\frac{2\beta}{3}}$ and large enough $t$, which follows using union bound. 
This means that the probability of the event that $|Y(\tau)|$ does not cross $\tau^{\frac{1}{2}-\beta}$ after $\tau=t$ approaches $1$ fast as $t\to \infty$. This, in turn, implies that $Y(t)$ almost surely remains within an envelope of the shape $t^{\frac{1}{2}-\beta}$ for $\beta>0$.

% Further, this implies that for $G(x) \gtrsim \frac{1}{x^{\delta}}$, we have that for any $0<c'<\delta$,
% \[\mathbb{P}_0(\cup_{\tau=t}^\infty \{|Y(t)| \ge c~\tau^{c'}) \le a \exp(-b~t^{\epsilon}).\]
% This implies that with high probability the opinion difference remains within an envelope that grows slower than any polynomial.

The opinion difference process $Y(t)$ is a Markov process \cite{Baccelli2017}. A high probability bound of the above kind is not uncommon for well-behaved Markov processes. %%%%%%%%%%%In fact, Markov chains with fast decaying stationary distribution would generally result in such bounds. 
However,  we note that $Y(t)$ is structurally quite different from the Markov chains observed in areas such as queuing systems and population dynamics. 

The Markov process $Y(t)$ lies in the class of asymptotically drift zero Markov processes, i.e., its expected conditional drift given $Y(t)=y$ tends to zero as $|y|\to \infty$. This is because the probability of influence decays with increasing opinion differences. Note that even the stable dynamics has asymptotically zero drift. 
Furthermore, unlike queuing processes, $|Y(t)|$ has unbounded jumps towards zero. 

Though the phenomenon of asymptotic zero drift is in direct opposition to strong concentration, unbounded jumps towards zero are conducive to concentration. However, the jumps are probabilistic and the probability decays with increasing opinion difference. So, it is not intuitively clear whether $Y(t)$ concentrates or not. In that sense,  Proposition~\ref{prop:mainresult} settles this dilemma affirmatively for stable SBC dynamics. For other SBC dynamics that are not stable, $Y(t)$ is either a null recurrent or a transient Markov process, and hence, is not expected to concentrate.

\section{Proof Outline of Theorem~\ref{theorem_sg}}
\label{sec:proofOutline}
In this section, we outline the proof of Theorem~\ref{theorem_sg} and discuss the main intuitions behind them. These intuitions are adapted and extended for obtaining  the results on bistar social graph. We start with a relatively simpler case of bounded $\tilde{n}(t)$ and discuss the basics of our proof technique in this context. Then, we extend that proof technique to sub-Gaussian noise in Proposition~\ref{prop:mainresult}.

\subsection{Bounded Noise: A Useful First Step}
\label{sec:boundedN}
As discussed above, the opinion difference process $Y(t)$ is a Markov chain \cite{Baccelli2017}
with unbounded state-dependent jumps and asymptotically zero drift. As a result, its analysis is intricate than the usual Markov chains \cite{Baccelli2017}. So, towards proving the main result in Theorem~\ref{theorem_sg}, as a first step, we consider a relatively simpler setting: $\tilde{n}(t)$ is bounded, i.e., it has a support $[-D, D]$ for some $D>0$ and $G(x) \gtrsim \frac{1}{x^{1-\delta}}$. 

\begin{theorem}\label{theorem_bounded}
Let $k=c~t^{\frac{1}{2}-\beta}$ for $c, \beta>0$ . Consider a two-agent stable dynamics with $G(x) \gtrsim \frac{1}{x^{1-\delta}}$ for some $\delta>0$, and bounded noise model. Then,  for all $t\ge0$, \[\mathbb{P}_0(\lvert Y(t) \rvert \geq k) \leq 2 \Big (\frac{t}{1-G(Dt)}+1 \Big ) \exp{\Big(-\frac{\sqrt{2 G(Dt)}}{D} k\Big)}.\]
\end{theorem}

%For $G=0$ and $Y(0)=0$, $Y(t)$ is the sum of i.i.d. zero mean noise, i.e., $\mathbb{P}(c_1 \sqrt{t} \le |Y(t)| \le c_2 \sqrt{t}|Y(0)=0)=1-o(1)$ for some $c_1, c_2>0$. However, for $G(x) \gtrsim \frac{1}{|x|^{1-\delta}}$, $Y(t)$ has a much stronger concentration around $0$.
The following simpler bound can be obtained from Theorem~\ref{theorem_sg}, for all sufficiently large $t$.

\begin{proposition}
\label{prop:boundedSimpleStatement}
Consider a two-agent stable dynamics with $G(x) \gtrsim \frac{1}{x^{1-\delta}}$ for some $\delta>0$, and bounded noise model with the assumed characteristics. Let $k=c~t^{\frac{1}{2}-\beta}$ for $c, \beta>0$. Then, for $\theta>0$, $c_1\le \frac{4-2G_0}{1-G_0}$ and $\epsilon\le\frac{\delta}{2}-\beta$,
\[\mathbb{P}_0(|Y(t)| \ge k) \leq c_1 t \exp{(-\theta t^\epsilon)}\] for all $t>0$ such that $\sqrt{t^{1-2\beta-2\epsilon}G(Dt)}\ge \frac{\theta D}{\sqrt{2}c}$.
In addition, for $t>\Big(\frac{1}{\theta~\epsilon}\Big)^\frac{1}{\epsilon}$ and $c_2>0$, 
\[\mathbb{P}_0(\lvert Y(t) \rvert \geq k)\leq c_1 \exp{(-c_2 t^\epsilon)}.\]
%{\color{red}Get bounds on $c_2$ and $c_1$ in terms of $\theta$}
\end{proposition}

This is a direct consequence of Theorem~\ref{theorem_bounded}. As we discuss later, the ingredients in the proof of Theorem~\ref{theorem_bounded} are also vital in proving Theorem~\ref{theorem_sg} and the bound for bistar graph in Proposition~\ref{prop:extendedstar}.

Details of the proof of this result is presented in Appendix~\ref{proof_bound_1}. Here, we only discuss the main intuitions that are also useful in proving Proposition~\ref{prop:mainresult}.  The main part of the proof involves obtaining a suitable upper bound on $\mathbb{E}_0[e^{\lambda  Y(t) }]$, where $\mathbb{E}_0[\cdot]$ represents $\mathbb{E}[\cdot|Y(0)=0]$. The bound on $\mathbb{E}_0[e^{\lambda  Y(t) }]$ is obtained using three main steps.

First, for any $t\ge 0$, obtain an upper bound on $\mathbb{E}_0[e^{\lambda  Y(t+1) }]$ in terms of $\mathbb{E}_0[e^{\lambda  Y(t) }]$, $\mathbb{E}[e^{\lambda  \tilde{n}(t) }]$ and $G(|Y(t)|)$:
\begin{align*}
    \mathbb{E}_0[e^{\lambda  Y(t+1) }]
    &\leq \mathbb{E}[e^{\lambda  \tilde{n}(t) }]\Big(\mathbb{E}_0[e^{\lambda Y(t)}(1-G(|Y(t)|))]+1\Big).
\end{align*}

Second, note that as the noise is bounded, $|Y(\tau)|\le D \tau$ for any $\tau$. Using this fact along with the recursive relation, for any $t>0$, we obtain a bound on $\mathbb{E}_0[e^{\lambda  Y(t) }]$ as a sum of products of the moment generating function of noise at $\lambda$, $M_{\tilde{n}}(\lambda)$, and $\{G(D i): 0\le i\le t\}$:
\begin{align*}
    \mathbb{E}_0[e^{\lambda Y(t)}]&\leq M_{\tilde{n}}(\lambda) + \prod_{i=0}^{t-1}M_{\tilde{n}} (\lambda)(1-G(Di))\nonumber\\&+\sum \limits_{i=0}^{t-2} M_{\tilde{n}}(\lambda)^{t-i} \prod_{j=0}^{t-2-i}(1-G(D(t-1-j)).\nonumber
\end{align*}

Third, by using Hoeffding's Lemma \cite[Sec. 2.3]{lugosi} for the moment generating function of noise, $M_{\tilde{n}}(\lambda)$, we obtain the following bound on $\mathbb{E}_0[e^{\lambda  Y(t) }]$ in terms of a sum of products involving only $\lambda$ and $G(Dt)$:  
\begin{align*}
    \bar{\gamma}(\lambda)\sum \limits_{i=0}^{t-1} \Big[ \bar{\gamma}(\lambda)(1-G(Dt))\Big]^{i}\nonumber+\Big[\bar{\gamma}(\lambda)(1-G(Dt))\Big]^t.\nonumber
\end{align*}
where $\bar{\gamma}(\lambda)=\frac{1}{1-\frac{\lambda^2 D^2}{2}}$.

Finally, by the use of the Chernoff bound and a suitable choice of $\lambda$, the bound in Theorem~\ref{theorem_bounded} follows. 

The above proof technique has an interesting aspect to it: it starts with an initial weak bound on $|Y(t)|$ and refines it to a much stronger bound. This is also the case in the proof of Proposition~\ref{prop:mainresult}. However, in the latter case, a tighter initial weak bound is utilized, which leads to a tighter final bound. As a result,  Proposition~\ref{prop:mainresult} applies to all influence functions for which the two-agent SBC dynamics is stable.

\subsection{Extending to Sub-Gaussian Noise}
\label{sec:subG}

Proof of Theorem~\ref{theorem_sg} builds on the three-step approach discussed in Sec.~\ref{sec:boundedN}. However, proof technique for bounded noise cannot be directly extended to sub-Gaussian noise.
An important ingredient in the proof for the bounded noise case was the bound on $G(|Y(t)|)$ in the second step, which used the fact that $|Y(\tau)|\le D\tau$ for any $\tau$. Clearly, this is not true when the noise is sub-Gaussian. We circumvent this issue by introducing a high probability bound on $|Y(t)|$ instead of a deterministic bound, and by adapting the subsequent proof steps and the final step involving the Chernoff bound according to that high probability bound.

Towards that, we define events $A_t=\bigcap\limits_{\tau=h(t)}^t \{|Y(\tau)| \leq d_\tau\}$ where $h(t)<t$ and $d_{\tau}=D~\tau^{\frac{1}{2}+\beta'}$ for some $\beta'>0$.  %$k=c~t^{\frac{1}{2}-\beta}$ for some $c, \beta>0$.

We first discuss the final step involving the Chernoff bound since that would place the changes we make to the three preceding steps in the proof of Theorem~\ref{theorem_bounded} in the right perspective. 

In the proof of Theorem~\ref{theorem_bounded}, the natural Chernoff bound is 
\[\mathbb{P}_0(Y(t)\ge k) \le \mathbb{E}_0[e^{\lambda  Y(t) }] \exp{(-\lambda k)}.\]

In the current setting, we adopt the following useful modification.
\begin{align}
    \mathbb{P}_0(Y(t)\ge k|A_{t-1})  &\le \mathbb{E}_0[e^{\lambda  Y(t) }| A_{t-1}] \exp{(-\lambda k)} \nonumber \\
    \mathbb{P}_0(Y(t)\ge k) & \le \mathbb{P}_0(Y(t)\ge k|A_{t-1})+ 1 - \mathbb{P}_0(A_{t-1}), \label{eq:totalProb}
\end{align}
where, $A_{t-1}$ is an event defined in terms of $\{Y(\tau): 0\le \tau \le t-1\}$. The choice of the event $A_t$ is dictated by the fact that $|Y(t)|$ should be bounded on $A_t$ with high probability and  $1-\mathbb{P}_0(A_{t-1})$ should be rapidly approaching $0$ as $t\to\infty$.

Note that, clearly, $|Y(t)|$ is stochastically smaller than the absolute value of the sum of i.i.d. sub-Gaussian noise $|\sum_{\tau=1}^t \tilde{n}(\tau)|$. Thus, $\mathbb{P}_0(|Y(t)|\ge D~ t^{\frac{1}{2}+\beta'})$
is no more than
\[\mathbb{P}(|\sum_{\tau=1}^t \tilde{n}(\tau)|\ge D~t^{\frac{1}{2}+\beta'}) \le 2~\exp\left(-\frac{D^2}{2 \sigma^2} t^{2\beta'}\right),\]
which fast approaches $0$ as $t\to\infty$. Here, the bound on $\mathbb{P}(|\sum_{\tau=1}^t \tilde{n}(\tau)|\ge D~ t^{\frac{1}{2}+\beta'})$ follows from the concentration inequality for the sum of i.i.d. sub-Gaussian random variables.

Based on the above observations, while keeping the adaptation of the final step involving the Chernoff bound in mind, we adapt the three main steps from the proof of Theorem~\ref{theorem_bounded} as follows.

First, define $A_t=\bigcap\limits_{\tau=h(t)}^t\{|Y(\tau)|\le D~ \tau^{\frac{1}{2}+\beta'}\}$ and obtain an upper bound on $\mathbb{E}_0[e^{\lambda  Y(t+1) }| A_t]$ in terms of $\mathbb{E}_0[e^{\lambda  Y(t)}|A_t]$, $\mathbb{E}[e^{\lambda  \tilde{n}(t) }]$ and $G(|Y(t)|)$. Then, using the fact that $A_t$ imposes a symmetric constraint on $Y(t)$ and $Y(t)$ has a symmetric distribution, we show that
$\mathbb{E}_0[e^{\lambda  Y(t)}|A_t] \le \mathbb{E}_0[e^{\lambda  Y(t)}|A_{t-1}]$. This gives a recursive upper bound on  $\mathbb{E}_0[e^{\lambda  Y(t+1) }| A_t]$ in terms of $\mathbb{E}_0[e^{\lambda  Y(t)}|A_{t-1}]$, $\mathbb{E}[e^{\lambda  \tilde{n}(t) }]$ and $G(|Y(t)|)$.

Second, we use the fact that $|Y(t)| \le D~t^{\frac{1}{2}+\beta'}$ on the event $A_t$  to obtain a bound on $\mathbb{E}_0[e^{\lambda  Y(t) }|A_{t-1}]$ as a sum of products of the moment generating function of noise at $\lambda$ and $\{G(D~i^{\frac{1}{2}+\beta'}): 0\le i\le t\}$. 

Third, using the sub-Gaussian bound on the moment generating function of noise, we obtain a bound on $\mathbb{E}_0[e^{\lambda  Y(t) }|A_{t-1}]$ in terms of a sum of products involving only $\lambda$ and $G(D~ t^{\frac{1}{2}+\beta'})$. 

Finally, upper bound on the probability of $A_t$ is obtained using sub-Gaussian concentration and all the bounds are plugged into \eqref{eq:totalProb}.
A detailed proof that essentially formalizes the above steps is presented in Appendix \ref{proof_bound_sg}. 

The bound in Proposition~\ref{prop:mainresult} applies to the whole range of influence functions for which dynamics is stable. On the other hand, the bound in Proposition~\ref{prop:boundedSimpleStatement} applies to influence functions satisfying $G(x)\gtrsim \frac{1}{x^{1-\delta}} $. The main reason behind the improvement from Proposition~\ref{prop:boundedSimpleStatement} to Proposition~\ref{prop:mainresult} is the change in the second step of the proof of Theorem~\ref{theorem_bounded}. In the proof for the sub-Gaussian case, the initial bound on $|Y(t)|$ in the second step of the proof of Theorem~\ref{theorem_sg}, though probabilistic, is a significantly tighter bound, which results in a much better final bound. 

\section{SBC Dynamics on a Bistar Social Graph}\label{section:network}

\begin{figure}[h]
    \centering
 \begin{tikzpicture}
    \draw[-] (4,0) -- (4.5,0);
    \node at (4.7,0) {${G}$};
    \draw[-,dashed] (4,-0.4) -- (4.5,-0.4);
    \node at (4.7,-0.4) {$\tilde{G}$};
    \draw (0,0) circle [radius = 0.2] node {$1$};
    \draw[-] (0.2,0) -- (1.2,0);
    \draw (1.4,0) circle [radius = 0.2] node {$2$};
    \draw (-1.2,0) circle [radius = 0.2];
    \draw (-0.9,-0.9) circle [radius = 0.2];
    \draw (0,-1.2) circle [radius = 0.2];
    \draw[-,dashed] (-0.2,0) -- (-1,0);
    \draw[-,dashed] (0,-0.2) -- (0,-1);
    \draw[-,dashed] (-0.2,-0.2) -- (-0.8,-0.8);
    
    \draw (2.6,0) circle [radius = 0.2];
    \draw[-,dashed] (1.6,0) -- (2.4,0);
    \draw (1.4,-1.2) circle [radius = 0.2];
    \draw[-,dashed] (1.4,-0.2) -- (1.4,-1);
    \draw (2.3,-0.9) circle [radius = 0.2];
    \draw[-,dashed] (1.6,-0.2) -- (2.2,-0.8);
    
\end{tikzpicture}
    \caption{An illustration of multi-agent systems on a bistar graph}%
\label{fig:MA}
\end{figure}
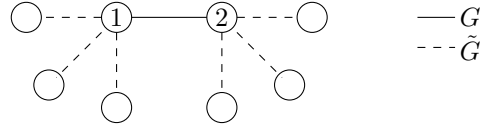

In this section, we utilize the insights obtained from two-agent dynamics and characterize opinion differences in multi-agent systems. Consider the bistar social graph in Fig.~\ref{fig:MA}, in which two agents,  $1$ and $2$, share an edge. There are two disjoint sets of agents $\mathcal{F}_1$ and $\mathcal{F}_2$, which share edges with agent $1$ and $2$, respectively, but share no edges among themselves. For any $f \in \mathcal{F}_1\cup\mathcal{F}_2$, $G_{f1}=G_{f2}=0$, i.e., agents $1$ and $2$ are not influenced by agents $\mathcal{F}_1\cup\mathcal{F}_2$. 
However, agents $1$ and $2$ are influenced by each other and their interactions follow the two-agent dynamics with influence function $G$, as discussed before. 

Agents $1$ and $2$ influence agents $\mathcal{F}_1$ and $\mathcal{F}_2$, respectively with an influence function $\tilde{G}(\cdot)$. However, $G_{1f}=0$ for all $f \in \mathcal{F}_2$ and $G_{2f}=0$ for all $f \in \mathcal{F}_1$.

If the opinion of $f\in \mathcal{F}_1$ at time $t$ is $X_f(t)$, then it updates as $X_1(t)+n_f(t)$ with probability $\tilde{G}(|Y_{f1}(t)|)$ and $X_f(t)+n_f(t)$ with probability $1-\tilde{G}(|Y_{f1}(t)|)$. Here, $n_f(t)$ captures noisy opinion updates at $t$. $Y_{f1}(t)=X_f(t)-X_1(t)$ represents the opinion difference between agent $1$ and its follower $f$. A similar dynamics happen for agents in $\mathcal{F}_2$ under the influence of agent $2$. 

The above dynamics represents a model of leader-follower societies with two leaders or parties of a bipartisan democratic polity. This SBC dynamics is stable for $G(x)\gtrsim \frac{1}{x^{\alpha}}$ for $\alpha<2$ and $\tilde{G}(x)\gtrsim \frac{1}{x^{\alpha'}}$ for $\alpha'<2$ \cite[Theorem III.6]{Baccelli-Infocom}.  
The main result of this section is a concentration bound on the opinion difference between a leader and their follower in such stable dynamics at a finite time. We assume that the initial opinion difference between the agents is zero.

Although the dynamics is related to the two agents, its behavior is quite different from that of the two-agent dynamics. This can be observed by following the dynamics of $Y_{f1}(t)$ for $f\in \mathcal{F}_1$. In two-agent dynamics, $Y_{f1}(t)$ evolves due to noise or $|Y_{f1}|$ jumps closer to $0$ if $f$ is  influenced by $1$. However, as agent $2$ can influence agent $1$ and pull it closer to itself, $|Y_{f1}|$ can experience jump increments. As a result, analyzing this dynamics is much more challenging than the two-agent dynamics. Here, by extending the proof techniques for two-agent dynamics, we obtain bounds on opinion differences for a subset of $G$ and $\tilde{G}$ and bounded noise. Define $G_0:=G(0)$ and $\Tilde{G}_0:=\Tilde{G}(0)$.

\begin{proposition}
\label{prop:extendedstar}
Consider multi-agent stable dynamics with bounded noise as described above. Assume $G(x)=\frac{G_0}{1+x^{1-\delta}}$ for some $\delta>0$ and $0<G_0\le1$, and $\tilde{G}(x)\gtrsim \Big(\frac{1}{x}\Big)^{\frac{2}{3}-\tilde{\delta}}$ for some $\tilde{\delta}>0$. Define $\tilde{k}=c~t^{\frac{1}{2}-\tilde{\beta}}$ for $c, \tilde{\beta}>0$. Then, for $t>\Big(\frac{2}{\theta\varsigma}\Big)^\frac{1}{\varsigma}$ and  positive (independent of $t$) constants $\theta$, $c_1\le 2\max \{\frac{4-2G_0}{1-G_0},\frac{4-2\tilde{G}_0}{1-\Tilde{G}_0}\}$ and $c_2>0$,
\begin{align*}
    \mathbb{P}_0(|Y_{f1}(t)|\ge \tilde{k})\le c_1 \exp{(-c_2 t^{\varsigma})}
\end{align*}
where $\varsigma=\min\{\epsilon\xi, \beta-\tilde{\beta}\}$, $0<\xi<1-2\beta$ and $\epsilon\le \frac{\delta}{2}-\beta$.
%{\color{red}Get bounds on $c_2$ and $c_1$ in terms of $\theta$}
\end{proposition}

Note that the SBC dynamics is known to be stable for $G(x)\gtrsim \frac{1}{x^{\alpha}}$ for $\alpha<2$ and $\tilde{G}(x)\gtrsim \frac{1}{x^{\alpha'}}$ for $\alpha'<2$. In that sense, the above bound applies to a \emph{subset} of stable dynamics, with very strong influence functions \cite{Baccelli2015,Baccelli2017}. Based on numerical evidence, we conjecture that a qualitative similar concentration bound holds for all stable influence functions $G(\cdot)$ and $\tilde{G}(\cdot)$. However, the proof (in Appendix~\ref{proof:extendedstar}), while already challenging, does not extend directly to cover all stable dynamics between the leaders and followers.

Proposition~\ref{prop:extendedstar} is a direct consequence of the following bound on $|Y_{f1}(t)|$ at any $t$. Define $l(t)=t^\xi$ for $0<\xi<1$, $\bar{d}_\tau=D~\tau^{\frac{1}{2}-\beta}$ for some $\beta>0$, and $\tilde{d}_t=2D~t^{\frac{3}{2}}$. 
\begin{theorem}
\label{theorem:extendedstar}
Assume $G(x)=\frac{G_0}{1+x^{1-\delta}}$ for some $\delta>0$ and $0<G_0\le1$. With $\tilde{G}(x)\gtrsim \Big(\frac{1}{x}\Big)^{\frac{2}{3}-\tilde{\delta}}$ for some $\tilde{\delta}>0$, $\lambda=\frac{2\ln 2}{\bar{d}_t}$, and positive (independent of $t$) constants $\theta$, $c_1\le \frac{4-2G_0}{1-G_0}$,
\begin{align*}
    \mathbb{P}_0(|Y_{f1}(t)|&\ge \tilde{k})\le c_1t^2\exp{(-\theta l(t)^\epsilon)}\\
    &+2\Big(\frac{t-l(t)}{1-\tilde{G}(\tilde{d}_t)}+\exp{\Big(\frac{5}{8}\lambda^2 D^2 l(t)\Big)}\Big)\exp{(-\lambda \tilde{k})}
\end{align*}
for all $t$ such that $\sqrt{t^{1-2\beta-2\epsilon}G(Dt)}\ge \frac{\theta}{\sqrt{2}}$ where $\epsilon\le \frac{\delta}{2}-\beta$. 
\end{theorem}

Note that the concentration results in Theorem~\ref{theorem:extendedstar} and Proposition~\ref{prop:extendedstar} assume a stronger influence between a leader and its followers than between the two leaders. This assumption, while stronger than needed for stability, seems reasonable in reality: i.e., a leader influences its followers more than it influences a rival. In our analysis, a more substantial influence by a leader on its follower aids the concentration of their opinion difference despite the perturbations caused by a rival.

The proof technique of Theorem~\ref{theorem:extendedstar} exploits the high probability bound on the opinion difference between agents $1$ and $2$ (obtained in Theorem~\ref{theorem_bounded}) to upper-bound $Y(\tau)$ in the second step. We define an event $B_t=\bigcap\limits_{\tau=l(t)}^t\{|Y(\tau)|\le \bar{d}_\tau\}$ where $l(t)=t^{\xi}$ for some $\xi>0$. Then, we adapt a modification of the Chernoff bound to get the required bound.
\begin{align*}
    \mathbb{P}_0(Y_{f1}(t)\ge \tilde{k}|B_{t-1}) &\le \mathbb{E}_0[e^{\lambda  {Y_{f1}(t)}}| B_{t-1}] \exp{(-\lambda \tilde{k})} \nonumber \\
    \mathbb{P}_0(Y_{f1}(t)\ge \tilde{k})&\le \mathbb{P}_0(Y_{f1}(t)\ge \tilde{k}|B_{t-1})+ \mathbb{P}_0(B_{t-1}^C).
\end{align*}
Details are available in Appendix~\ref{proof:extendedstar}. 
By a similar interpretation of notations, Proposition~\ref{prop:extendedstar} also characterizes $Y_{g2}(t)$, the opinion difference between agent $2$ and its follower $g\in \mathcal{F}_2$.
Indeed, by using triangle inequality, union bound, and the high probability bounds on $Y(t)$, $Y_{f1}(t)$ and $Y_{g2}(t)$, we derive similar concentration bounds on the cross opinion differences, i.e., opinion differences between follower $f\in\mathcal{F}_1$ and agent $2$, $g\in\mathcal{F}_2$ and agent $1$, and followers $f\in \mathcal{F}_1$ and $g\in\mathcal{F}_2$. In particular, the following proposition captures the opinion dynamics between two opposing groups of people in a society.

\begin{proposition}\label{prop:opposing}
Consider multi-agent stable dynamics with bounded noise as described above. Assume $G(x)=\frac{G_0}{1+x^{1-\delta}}$ for some $\delta>0$ and $0<G_0\le1$, and $\tilde{G}(x)\gtrsim \Big(\frac{1}{x}\Big)^{\frac{2}{3}-\tilde{\delta}}$ for some $\tilde{\delta}>0$. Define $\tilde{k}=c~t^{\frac{1}{2}-\tilde{\beta}}$ for $c, \tilde{\beta}>0$. Then, for $t>\Big(\frac{2}{\theta\varsigma}\Big)^\frac{1}{\varsigma}$ and  positive (independent of $t$) constants $\theta$, $c_1\le 5\max \{\frac{4-2G_0}{1-G_0},\frac{4-2\tilde{G}_0}{1-\Tilde{G}_0}\}$ and $c_2>0$,
\begin{align*}
    \mathbb{P}_0(|Y_{fg}(t)|\ge \tilde{k})\le c_1 \exp{(-c_2 t^{\varsigma})}
\end{align*}
where $\varsigma=\min\{\epsilon\xi, \beta-\tilde{\beta}\}$, $0<\xi<1-2\beta$ and $\epsilon\le \frac{\delta}{2}-\beta$.
\end{proposition}
We infer from Propositions~\ref{prop:extendedstar} and \ref{prop:opposing} that, in stable dynamics, the opinion difference between a follower and its leader has a stronger concentration than that between the followers of rival groups. Indeed, this is the case in reality.

An essential consideration in the proof techniques is the zero initial opinion difference between a pair of agents. This consideration plays a significant role in establishing symmetric constraints on the opinion difference process while providing analytical guarantees. In general, concentration occurs in stable dynamics (discussed in Sections~\ref{sec:mainResult} and \ref{section:network}) even with non-zero initial opinion differences. However, the rate of concentration will be slower.

\section{Simulations}\label{simulations}

This section presents numerical results to further elucidate the concentration of opinion differences in stable dynamics. We assume that noise in opinion differences follows uniform distribution in $[-D, D]$ for $D=20$.
\vspace{-0.3cm}
\subsection{Opinion dynamics of two agents}\label{fig:abs-two}
We consider a two-agent stable SBC dynamics with $G(x)=\frac{1}{1+x^{1-\delta}}$ for some $\delta>0$.
We infer from the opinion difference model in Section~\ref{sec:mainResult} that $G=0$ represents a symmetric random walk. In that case, for noise with a finite second moment, $\frac{D^2}{3}$, the expected drift in opinion difference after $t$ units of time is $\frac{D}{\sqrt{3}}\sqrt{t}$ \cite[Chap. 2]{lugosi}. Therefore, with $G>0$, it is suggested that the expected drift is $c~t^{\frac{1}{2}-\beta}$ for $c, \beta>0$.

Fig.~\ref{fig:a} illustrates the evolution of opinion difference between the two agents with respect to $\delta$ for $\delta=0.2, 0.5, 0.8$ and $\beta=\frac{\delta}{4}$. $\delta$ characterizes the strength of influence among the agents, i.e., the larger the $\delta$, the more likely the agents influence each other despite large opinion differences. In other words, as $\delta$ increases, their opinion difference likely stays within the envelope of the form $t^{\frac{1}{2}-\beta}$. 
%the empirical probability that the opinion difference at time $t$, $Y(t)$, goes beyond the envelope of the form $t^{\frac{1}{2}-\beta}$ based on the strength of influence among the agents. We assume that $G(x)=\frac{1}{1+x^{1-\delta}}$ and  $\beta=\frac{\delta}{4}$ for some $\delta>0$. We plot the empirical probability for $\delta=0.2, 0.5,$ and $0.8$. The parameter $\delta$ characterizes the strength of influence by the agents. That is, larger the $\delta$, heavier is the tail of $G(\cdot)$ and therefore, more likely the agents influence each other despite large opinion differences. In other words, as $\delta$ increases, the likelihood that the opinion difference stays within the envelope of the form $t^{\frac{1}{2}-\beta}$ increases. 

With $\delta=0.5$, we depict a concentration bound for $Y(t)$ (seen in Proposition~\ref{prop:boundedSimpleStatement}) in Fig.~\ref{fig:c} for $c_1=1$ and $c_2=0.5$.
\begin{figure}[hbpt!]
\centering
\captionsetup{justification=centering}
\subfloat[Evolution of opinion difference with $G(x)\gtrsim\frac{1}{x^{1-\delta}}$ for $\delta=0.2, 0.5$ and $0.8$.\label{fig:a}]{
\includegraphics[width=0.4\textwidth]{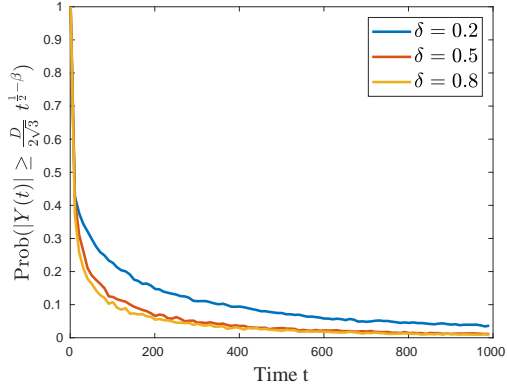}
}\\
\subfloat[Concentration bounds for two-agent dynamics with $G(x)\gtrsim\frac{1}{x^{1-\delta}}$ for $\delta=0.5$.\label{fig:c}]{
\includegraphics[width=0.4\textwidth]{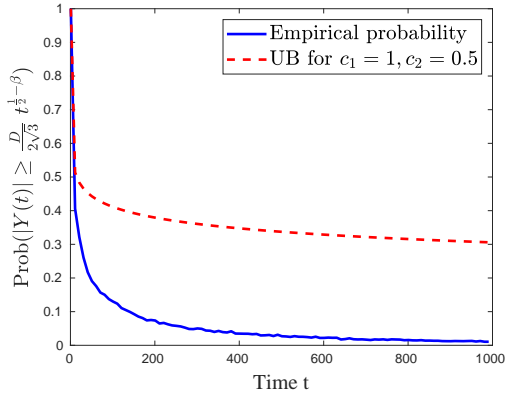}
}\\

\caption{Illustration of the behaviour of two-agent stable SBC dynamics in finite time.}%
\label{fig:wrtdelta}
\end{figure} 

\vspace{-0.3cm}
\subsection{Opinion dynamics on bistar graphs}

\begin{figure}[hbpt!]
\centering
\captionsetup{justification=centering}
\subfloat[Opinion difference between a follower and its leader.\label{fig:network}]{
\includegraphics[width=0.4\textwidth]{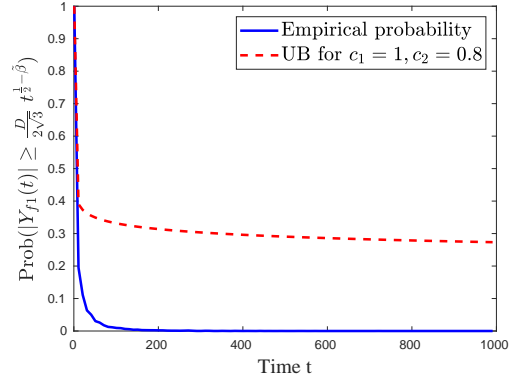}
}\\
\subfloat[Opinion difference between followers of rival groups.\label{fig:network2}]{
\includegraphics[width=0.4\textwidth]{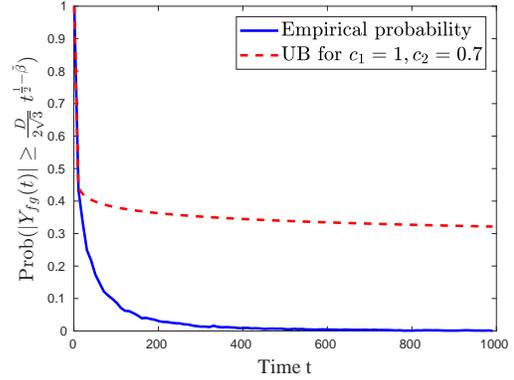}
}\\

\caption{An illustration of concentration bounds in stable multi-agent SBC dynamics.}%
\label{fig:boundsnet}
\end{figure}
For the multi-agent dynamics discussed in Section~\ref{section:network}, we consider that $G(x)=\frac{1}{1+x^{1-\delta}}$ for $\delta=0.5$ and $\tilde{G}(x)=\frac{1}{1+x^{\frac{2}{3}-\tilde{\delta}}}$ for $\tilde{\delta}=0.55$. We choose $\beta=\frac{\delta}{4}$ and $\tilde{\beta}=\frac{\tilde{\delta}}{10}$. In Fig.~\ref{fig:network}, we illustrate how the opinion difference between a leader and their follower evolves when there is a perturbation in the leader's opinion due to a rival leader. In the same figure, we depict our concentration bound in Proposition~\ref{prop:extendedstar} for constants $c_1=1$ and $c_2=0.8$. Likewise, Fig.~\ref{fig:network2} illustrates the concentration of opinion difference between followers $f\in \mathcal{F}_1$ and $g\in \mathcal{F}_2$.

These numerical results suggest that the actual concentration of opinion differences in stable dynamics is much stronger. Furthermore, the analytical bounds hold for all $t$ for suitable choices of $c_1$ and $c_2$ (Fig.~\ref{fig:c} and Fig.~\ref{fig:boundsnet}). 
%Although it is suggested that the analytical bounds in all the propositions hold only for large $t$, it is apparent from figures \ref{fig:c} and \ref{fig:boundsnet} that the bounds hold for all $t$ for suitable choices of constants $c_1$ and $c_2$. In particular, these results demonstrate that the actual concentration of opinion differences is much stronger. The main motive of this work is to characterize opinion differences in tractable stable systems, and extend our analysis to more complicated network structures in the future. However, these numerical results suggest a scope for tightening the theoretical bounds for all $t$.

\section{Concluding Remarks}

In this work, we characterized the finite-time behavior of stable stochastic bounded confidence opinion dynamics on a bistar graph. Specifically, we obtained a high probability bound on the opinion difference between a pair of agents at a finite time. Our proof technique is based on bounding the conditional moment generating function of the opinion difference on a high probability subset of the sample space and adapting the Chernoff bound accordingly. The main motive of this work is to characterize opinion differences in tractable stable systems and extend our analysis to more complicated network structures in the future. Indeed, numerical results suggest scope for tightening the theoretical bounds.

\section*{Acknowledgement}
The first author's work was supported by the Prime Minister's Research Fellows (PMRF) scheme. The work of AC was partially supported by the Department of Science and Technology, Government of India, under Grant INSPIRE/04/2016/001171.

\appendices
\section{Proof of Theorem~\ref{theorem_bounded}}
\label{proof_bound_1}
\begin{proof}[\unskip\nopunct]

Assume that $\tilde{n}(t)$ is bounded in $[-D,D]$. Let $M_{\tilde{n}}(\lambda)$ denote its moment generating function at $\lambda$. By using Chernoff bound, for $\lambda>0$,
\begin{align}\label{cher}
    \mathbb{P}_0(\lvert Y(t) \rvert \geq k) = 2\mathbb{P}_0(Y(t) \geq k)  \leq 2\mathbb{E}_0[e^{\lambda  Y(t) }]\exp{(-\lambda k)}.
\end{align}
We now focus on the expectation term in (\ref{cher}). By the law of iterated expectations,
\begin{align*}
   \mathbb{E}_0[e^{\lambda  Y(t+1) }]& = \mathbb{E}_0[\mathbb{E}[e^{\lambda Y(t+1) } | Y(t)]]\\
   &=\mathbb{E}_0[e^{\lambda (Y(t)+\tilde{n}(t))}(1-G(|Y(t)|))\\
   &\qquad+e^{\lambda \tilde{n}(t)}G(|Y(t)|)].
   %&=\mathbb{E}[e^{\lambda \tilde{n}(t)}]\mathbb{E}[e^{\lambda Y(t)}(1-G(|Y(t)|)+G(|Y(t)|)]& \because(\text{for any $t>0$, $\tilde{n}(t)$ and $Y(t)$ are independent})\\
   %&\leq M_{\tilde{n}}(\lambda)\Big(\mathbb{E}[e^{\lambda Y(t)}](1-G(|Y(t)|))+1\Big) &\because(\text{$G(x) \leq 1 \forall x$})%  and in $[0,t]$, $G(|Y(t)|) \geq G(Dt)$})
\end{align*}
For any $t>0$, $\tilde{n}(t)$ and $Y(t)$ are independent. Therefore,
\begin{align*}
    \mathbb{E}_0[e^{\lambda  Y(t+1) }]& = \mathbb{E}[e^{\lambda \tilde{n}(t)}]\mathbb{E}_0[e^{\lambda Y(t)}(1-G(|Y(t)|))+G(|Y(t)|)]\\
    &\leq M_{\tilde{n}}(\lambda)\Big(\mathbb{E}_0[e^{\lambda Y(t)}(1-G(|Y(t)|))]+1\Big).
\end{align*}
As $G(\cdot)$ is decreasing in its argument, in time window $[0,t]$, $G(|Y(t)|) \geq G(Dt)$ for any $t$. We have
\begin{align*}
    \mathbb{E}_0[e^{\lambda  Y(t+1) }]\leq M_{\tilde{n}}(\lambda)\Big(\mathbb{E}_0[e^{\lambda Y(t)}](1-G(Dt))+1\Big).
\end{align*}
This recursive inequality can be expanded to obtain
\begin{align}
    \mathbb{E}_0[e^{\lambda Y(t)}]&\leq M_{\tilde{n}}(\lambda)\nonumber\\&+\sum \limits_{i=0}^{t-2} M_{\tilde{n}}(\lambda)^{t-i} \prod_{j=0}^{t-2-i}(1-G(D(t-1-j))\nonumber\\&+ \prod_{i=0}^{t-1}M_{\tilde{n}} (\lambda)(1-G(Di)).\label{rec_bounded}
    %&\leq  M_{\tilde{n}}(\lambda)+\sum \limits_{i=0}^{t-2} M_{\tilde{n}}(\lambda)^{t-i}(1-G(Dt))^{t-i-1}+\Big[M_{\tilde{n}} (\lambda)(1-G(Dt))\Big]^t\nonumber\\
    %& \leq M_{\tilde{n}}(\lambda)\sum \limits_{i=0}^{t-1} \Big[ M_{\tilde{n}}(\lambda)(1-G(Dt))\Big]^{i}+\Big[M_{\tilde{n}} (\lambda)(1-G(Dt))\Big]^t \label{c_rec}
\end{align}
We exploit the observation that $G(\cdot)$ is decreasing and simplify the RHS of (\ref{rec_bounded}) term by term. Focusing on the second term,
\begin{align*}
    \sum \limits_{i=0}^{t-2} M_{\tilde{n}}(\lambda)^{t-i} &\prod_{j=0}^{t-2-i}(1-G(D(t-1-j))\\&\leq \sum \limits_{i=0}^{t-2} M_{\tilde{n}}(\lambda)^{t-i} \prod_{j=0}^{t-2-i}(1-G(Dt))\\
    &= \sum \limits_{i=0}^{t-2} M_{\tilde{n}}(\lambda)^{t-i}(1-G(Dt))^{t-i-1}\\
    &= M_{\tilde{n}}(\lambda)\sum \limits_{i=1}^{t-1} \Big[ M_{\tilde{n}}(\lambda)(1-G(Dt))\Big]^{i}.
\end{align*}
Now, the third term: 
\begin{align*}
    \prod_{i=0}^{t-1}M_{\tilde{n}} (\lambda)(1-G(Di))
    &\leq M_{\tilde{n}} (\lambda)^t \prod_{i=0}^{t-1}(1-G(Dt))\\
    &= \Big[M_{\tilde{n}} (\lambda)(1-G(Dt))\Big]^t.
\end{align*}
Putting all these terms together in (\ref{rec_bounded}), we get
\begin{align}
    \mathbb{E}_0[e^{\lambda Y(t)}]&\leq M_{\tilde{n}}(\lambda)\sum \limits_{i=0}^{t-1} \Big[ M_{\tilde{n}}(\lambda)(1-G(Dt))\Big]^{i}\nonumber\\&+\Big[M_{\tilde{n}} (\lambda)(1-G(Dt))\Big]^t. \label{c_rec}
\end{align}
By Hoeffding's Lemma \cite[Sec. 2.3]{lugosi}, 
\begin{align*}
    M_{\tilde{n}}(\lambda) \leq \exp{\Big(\frac{\lambda^2D^2}{2}\Big)}. 
\end{align*}
We know that $\exp{(x)}\le \frac{1}{1-x}$ for all $x\ge0$. This means that $\exp{\Big(\frac{\lambda^2D^2}{2}\Big)}\le \frac{1}{1-\frac{\lambda^2 D^2}{2}}$. Define $\bar{\gamma}(\lambda)=\frac{1}{1-\frac{\lambda^2 D^2}{2}}$. Now, (\ref{c_rec}) becomes 
\begin{align*}
    \mathbb{E}_0[e^{\lambda Y(t)}]&\leq \bar{\gamma}(\lambda)\sum \limits_{i=0}^{t-1} \Big[ \bar{\gamma}(\lambda)(1-G(Dt))\Big]^{i}\nonumber\\&+\Big[\bar{\gamma} (\lambda)(1-G(Dt))\Big]^t.
\end{align*}
Therefore,
\begin{align}
     \mathbb{P}_0(\lvert Y(t) \rvert \geq k) &\leq 2 \Big (\bar{\gamma}(\lambda)\sum \limits_{i=0}^{t-1} \Big[ \bar{\gamma}(\lambda)(1-G(Dt))\Big]^{i}\nonumber\\&+\Big[\bar{\gamma} (\lambda)(1-G(Dt))\Big]^t \Big ) \exp{(-\lambda k)}. \label{c_bound_p}
\end{align}
Based on the stability criterion, we consider that $G(x) \gtrsim \frac{1}{x^{1-\delta}}$ for some $\delta>0$. 
We observe from (\ref{c_bound_p}) that, as $\lambda$  increases, the exponential term decreases, whereas $\bar{\gamma}( \lambda)$ increases. One way to choose an optimal $\lambda$ for a better bound is such that $\bar{\gamma}( \lambda) (1-G(Dt))\le1$. This implies that $\frac{\lambda^2 D^2}{2}\le G(Dt)$ yielding $\lambda \le \frac{\sqrt{2 G(Dt)}}{D}$. We assume $\lambda = \frac{\sqrt{2 G(Dt)}}{D}$. With this choice of $\lambda$, (\ref{c_bound_p}) can be rewritten as
\begin{align}
     \mathbb{P}_0(\lvert Y(t) \rvert \geq k) &\leq 2 \Big (\frac{t}{1-G(Dt)}+1 \Big ) \exp{\Big(-\frac{\sqrt{2 G(Dt)}}{D} k\Big)}. \label{c_bound_p1}
\end{align}
Note that (\ref{c_bound_p1}) holds for all $t\ge 0$. Clearly, the bound in (\ref{c_bound_p1}) depends on the functional form of $G(\cdot)$. In particular, (\ref{c_bound_p1}) can be written in a compact form as follows for all $t>0$ such that $\frac{\sqrt{2G(Dt)}}{D}k\ge \theta t^\epsilon$. For $\epsilon \le \frac{\delta}{2}-\beta$ and positive constants $c_1\le \frac{4-2G_0}{1-G_0}$ and $\theta$, 
\begin{align}
     \mathbb{P}_0(\lvert Y(t) \rvert \geq k) &\leq c_1 t \exp{(-\theta t^\epsilon)}. \label{c_bound_p2}
\end{align}
Further, for $t>\Big(\frac{1}{\theta~\epsilon}\Big)^\frac{1}{\epsilon}$ and $c_2>0$, (\ref{c_bound_p2}) can be refined into
\begin{align}
     \mathbb{P}_0(\lvert Y(t) \rvert \geq k) &\leq c_1 \exp{(-c_2 t^\epsilon)}. \label{c_bound_p3}
\end{align}

Recall that $k=c~t^{\frac{1}{2}-\beta}$ for  $c, \beta>0$. With $G(x) \gtrsim \frac{1}{x^{1-\delta}}$ for some $\delta>0$ and $\epsilon\le \frac{\delta}{2}-\beta$, we observe that (\ref{c_bound_p3}) decreases with $t$. This, in turn, is not only in accordance with the stability results in the literature but also proves our intuition for $k$. Clearly, this technique cannot be directly extended to the case where $G(x) \gtrsim \frac{1}{x^{2-\delta}}$ as it will be evident from the available choices of $\epsilon$.
\end{proof}
%{\color{red}For the final bound in \eqref{c_bound_p1}, you only need $\bar{\gamma}( \lambda) (1-G(Dt))\le 1$, not $<$.
%Do not use $\lambda= m ....$, here. Rather, take $\lambda=\sqrt{2 G(Dt)}/D$. Plug this in everywhere there is $\lambda$ and get the final bound in terms of $G(Dt)$, $\epsilon(t)$ and $k$. This bound will be true for all $t$. Then, find the condition on $t_0$, which will depend on $G(\cdot)$, beyond which  the compact form $\exp(-t^{\beta})$ will be valid.}

\section{Proof of Theorem~\ref{theorem_sg}}
\label{proof_bound_sg}
\begin{proof}[\unskip\nopunct]
We use the following Lemmas~\ref{conditional_P} and \ref{complement_sg} to prove Theorem~ \ref{theorem_sg}. Proof of the lemmas are detailed in Appendix~\ref{section:conditional_P} and \ref{section:complement_sg} respectively.

Recall the event $A_t=\bigcap\limits_{\tau=h(t)}^t \{|Y(\tau)| \leq d_\tau\}$ where $h(t)=t^\zeta, \; 0<\zeta<1$, and $d_{\tau}=D~\tau^{\frac{1}{2}+\beta'}$ for some $\beta'>0$.
\begin{lemma}\label{conditional_P}
With $G(x) \gtrsim \frac{1}{x^{2-\delta}}$ for some $\delta>0$,
\begin{align*}
    &\mathbb{P}_0(|Y(t)| \geq k | A_{t-1})\\ &\leq 2\Big[\frac{t-h(t)}{1-G(d_t)}+\exp{\Big(G(d_t)h(t)\Big)}\Big] \exp{\Big(- \frac{\sqrt{2 G(d_t)}}{\sigma}k\Big)}.
\end{align*}
\end{lemma}

\begin{lemma} \label{complement_sg}
Recall that $d_{\tau}=D~\tau^{\frac{1}{2}+\beta'}$ for some $\beta'>0$ and $D > 0$. Then, for $c'>0$,
\begin{align*}
    \mathbb{P}_0(A_{t-1}^\mathrm{C})\leq 2(t-h(t))\exp{(-c'h(t)^{2 \beta'})}.
\end{align*}
\end{lemma}
By law of total probability,
\begin{equation*}
    \begin{split}
    \mathbb{P}_0(|Y(t)| \geq k)&=\mathbb{P}_0(|Y(t)| \geq k|A_{t-1})\mathbb{P}_0(A_{t-1})+\\&\quad\; \mathbb{P}_0(|Y(t)| \geq k|A_{t-1}^\mathrm{C})\mathbb{P}_0(A_{t-1}^\mathrm{C})\\&\leq \mathbb{P}_0(|Y(t)| \geq k|A_{t-1})+\mathbb{P}_0(A_{t-1}^\mathrm{C}).
\end{split}
\end{equation*}

From Lemma~\ref{conditional_P} and Lemma~\ref{complement_sg},
\begin{align*}
     \mathbb{P}_0(|Y(t)| \geq k)&\leq 2(t-h(t))\exp{(-c'h(t)^{ 2\beta'})}+2\Big[\frac{t-h(t)}{1-G(d_t)}\nonumber\\ &+\exp{\Big(G(d_t)h(t)\Big)}\Big] \exp{\Big(- \frac{\sqrt{2 G(d_t)}}{\sigma}k\Big)}.
\end{align*}
Now, for $\zeta\le1-\frac{\delta}{2}$ and $\epsilon\le\frac{\delta}{6}-\frac{2\beta}{3}$, the final bound can be simplified to yield the result in Proposition~\ref{prop:mainresult}.
% as follows. For positive constants $\theta$ and $c_1\le \frac{2(3-2G_0)}{1-G_0}$ and all $t$ such that $\frac{\sqrt{2G(d_t)}}{\sigma}k\ge \theta t^\epsilon$,  
% \begin{align*}
%     \mathbb{P}_0(\lvert Y(t) \rvert \geq k)\leq c_1t\exp{(-\theta t^\epsilon)}.
% \end{align*}
%Furthermore, for $t>\Big(\frac{1}{\theta~\epsilon}\Big)^\frac{1}{\epsilon}$ and $c_2>0$, 
% \begin{align*}
%     \mathbb{P}_0(\lvert Y(t) \rvert \geq k)\leq c_1\exp{(-c_2t^\epsilon)}.
% \end{align*}
\end{proof}

\section{Proof of Lemma~\ref{conditional_P}}\label{section:conditional_P}
\begin{proof}[\unskip\nopunct]
The following Claims~\ref{restriction1} and \ref{ordering_sg} are useful to prove Lemma~\ref{conditional_P}.

\begin{claim}\label{restriction1}
Recall the event $A_t$. Then, for $\lambda>0$,
\begin{align*}
    \mathbb{E}_0[\exp{(\lambda Y(t))}|A_t] \leq\mathbb{E}_0[\exp{(\lambda Y(t))}|A_{t-1}]. 
\end{align*}
\end{claim}
%\subsection*{Proof of lemma \ref{restriction}}
\begin{proof}[Proof of Claim~\ref{restriction1}]
Let the conditional probability of $Y(t)$ given an arbitrary event $A$ be denoted as $f_{{Y(t)}|{A}}(.)$. Note that the constraints $\{\lvert Y(\tau) \rvert \leq d_\tau\}$ are symmetric for all $\tau$. With $Y(0)=0$, symmetric noise model and symmetric constraints, we observe that $f_{{Y(t)}|{A_{t}}}(.)$ and $f_{{Y(t)}|{A_{t-1}}}(.)$ are symmetric about zero. Hence,
\begin{align}
    \mathbb{E}_0[\exp{(\lambda Y(t))}|A_{t-1}]&= \int_0^\infty \exp{(\lambda y})f_{{Y(t)}|{A_{t-1}}}(y)dy\nonumber\\&\quad+\int_0^\infty \exp{(-\lambda y})f_{{Y(t)}|{A_{t-1}}}(y)dy\nonumber\\
    &=2 \int_0^\infty \cosh{(\lambda y})f_{{Y(t)}|{A_{t-1}}}(y)dy. \label{At-1}
\end{align}
Similarly, 
\begin{align}
    \mathbb{E}_0[\exp{(\lambda Y(t))}|A_{t}]=2 \int_0^\infty \cosh{(\lambda y})f_{{Y(t)}|{A_{t}}}(y)dy. \label{At}
\end{align}

Note that $\cosh{(\lambda y)}$ increases with $|y|$ for any $\lambda>0$ and  $f_{{Y(t)}|{A_{t}}}(.)$ is a restriction of $f_{{Y(t)}|{A_{t-1}}}(.)$ on smaller $|Y(t)|$. Therefore, from (\ref{At-1}) and (\ref{At}), we have the result.
\end{proof}

\begin{claim}\label{ordering_sg}
Let $\{Y'(t),t\ge 0\}$ be the process of opinion difference for a two-agent system with $G=0$. That is,
\begin{align*}
    Y'(t+1)=Y'(t)+\tilde{n}(t).
\end{align*}
Then, for any $t \ge 0$,
\begin{align*}
    \mathbb{E}_0[\exp{(\lambda Y(t))]} \leq \mathbb{E}_0[\exp{(\lambda Y'(t))]}\leq \exp{\Big(\frac{\lambda^2\sigma^2 t}{2}\Big)}.
\end{align*}
\end{claim}

%\subsection*{Proof of lemma \ref{ordering_sg}}
The proof of Claim~\ref{ordering_sg} uses the notion of stochastic ordering and so, we overview it below. For a random variable $X$, let $F_X$ and $\Bar{F}_X$ represent its distribution function and tail distribution respectively, i.e. for any $x\in \mathbb{R}$,
\begin{align*}
    F_X(x)=\mathbb{P}(X \le x),\\
    \bar{F}_X(x)=\mathbb{P}(X > x).
\end{align*}

\begin{definition}[{Stochastic Ordering \cite[Sec. 1.2]{Stoyan1983}}]\label{s_ordering}
Given two random variables $X$ and $Y$ taking values in $\mathbb{R}$, we denote $X \le_{st} Y$ if 
\begin{align*}
    F_X(l) \ge F_Y(l) \; \forall l \in \mathbb{R}
\end{align*}
or equivalently, if
\begin{align*}
    \bar{F}_X(l) \le \bar{F}_Y(l) \; \forall l \in \mathbb{R}.    
\end{align*}
Also, if $X \le_{st} Y$, then $\mathbb{E}[f(X)]\le\mathbb{E}[f(Y)]$ for all non-decreasing functions $f$ for which the expectations exist. 
\end{definition}

\begin{proof}[Proof of Claim~\ref{ordering_sg}]
Clearly, $Y(t) \leq_{st} Y'(t)$. As exponential function is non-decreasing for $\lambda>0$, by Definition~\ref{s_ordering},  $\mathbb{E}_0[\exp{(\lambda Y(t)) }]\leq \mathbb{E}_0[\exp{(\lambda Y'(t))}]$. Since $Y(0)=0$, $Y'(t)=\sum_{\tau=0}^{t-1} \tilde{n}(\tau).$
We also assumed that  $\tilde{n}(t) \in \mathcal{SG}(\sigma^2)$ for all $t$. Therefore, $Y'(t) \in \mathcal{SG}(\sigma^2 t)$. That is,
\begin{align*}
    \mathbb{E}_0[\exp{\lambda Y'(t)}]=\mathbb{E}_0[\exp\Big({\lambda  \sum_{\tau=0}^{t-1} \tilde{n}(\tau)}\Big)]\leq \exp{\Big(\frac{\lambda^2 \sigma^2 t}{2}\Big)}.
\end{align*}
\end{proof}

Along with these claims,  we use the proof technique of Theorem~\ref{theorem_bounded} to bound the conditional probability $\mathbb{P}_0(\lvert Y(t) \rvert \geq k|A_{t-1})$. By the law of iterated expectations,
\begin{align}
   \mathbb{E}_0[e^{\lambda  Y(t+1) }]& = \mathbb{E}_0[\mathbb{E}[e^{\lambda Y(t+1) } | Y(t)]]\nonumber\\
   & =M_{\tilde{n}}(\lambda)\mathbb{E}_0[e^{\lambda Y(t)}(1-G(|Y(t)|))+G(|Y(t)|)]. \label{unconditional_sg}
\end{align}
Now, from (\ref{unconditional_sg}),
\begin{align*}
    \mathbb{E}_0[e^{\lambda  Y(t+1)}|A_t]& =M_{\tilde{n}}(\lambda)\mathbb{E}_0[e^{\lambda Y(t)}(1-G(|Y(t)|))\\ &\quad\quad\quad\quad+G(|Y(t)|)|A_t]\\
    &\leq M_{\tilde{n}}(\lambda)\mathbb{E}_0[e^{\lambda Y(t)}(1-G(|Y(t)|))+1|A_t]\\
    & \labelrel\leq{g_dec} M_{\tilde{n}}(\lambda)(\mathbb{E}_0[e^{\lambda Y(t)}|A_t](1-G(d_t))+1)\\
    & \labelrel\leq{ineq_lemma} M_{\tilde{n}}(\lambda)(\mathbb{E}_0[e^{\lambda Y(t)}|A_{t-1}](1-G(d_t))+1), 
\end{align*}
where inequality~\eqref{g_dec} holds as $G(\cdot)$ is decreasing in its argument and inequality~\eqref{ineq_lemma} follows Claim~\ref{restriction1}.
So, we have a recursive inequality:
\begin{align*}
\mathbb{E}_0[e^{\lambda  
Y(t+1)}|A_t]\leq M_{\tilde{n}}(\lambda)(1+\mathbb{E}_0[e^{\lambda Y(t)}|A_{t-1}](1-G(d_t))),
\end{align*}
\begin{comment}
Expanding this recursive inequality, we get 
\begin{align}
    \mathbb{E}[e^{\lambda Y(t)}|A_{t-1}]&\leq M_{\tilde{n}}(\lambda)+\sum \limits_{i=h(t)}^{t-2} M_{\tilde{n}}(\lambda)^{t-i} \prod_{j=0}^{t-2-i}(1-G(d_{t-1-j}))+\mathbb{E}[\exp{(\lambda Y(h(t))}] \prod_{i=h(t)}^{t-1}M_{\tilde{n}} (\lambda)(1-G(d_i))\nonumber\\
    &\leq M_{\tilde{n}}(\lambda)+\sum \limits_{i=h(t)}^{t-2} M_{\tilde{n}}(\lambda)^{t-i} \prod_{j=0}^{t-2-i}(1-G(d_{t}))+\mathbb{E}[\exp{(\lambda Y(h(t))}] \prod_{i=h(t)}^{t-1}M_{\tilde{n}} (\lambda)(1-G(d_t))\nonumber\\    
    &\leq  M_{\tilde{n}}(\lambda)+\sum \limits_{i=h(t)}^{t-2} M_{\tilde{n}}(\lambda)^{t-i}(1-G(d_t))^{t-i-1}+\mathbb{E}[\exp{(\lambda Y(h(t))}] \Big[M_{\tilde{n}} (\lambda)(1-G(d_t))\Big]^{t-h(t)}\nonumber\\
    &\leq  M_{\tilde{n}}(\lambda)+M_{\tilde{n}}(\lambda)\sum \limits_{i=h(t)}^{t-2} M_{\tilde{n}}(\lambda)^{t-i-1}(1-G(d_t))^{t-i-1}+\mathbb{E}[\exp{(\lambda Y(h(t))}] \Big[M_{\tilde{n}} (\lambda)(1-G(d_t))\Big]^{t-h(t)}\nonumber\\
    & \leq M_{\tilde{n}}(\lambda)\sum \limits_{i=0}^{t-1-h(t)} \Big[ M_{\tilde{n}}(\lambda)(1-G(d_t))\Big]^{i}+\mathbb{E}[\exp{(\lambda Y(h(t))}] \Big[M_{\tilde{n}} (\lambda)(1-G(d_t))\Big]^{t-h(t)}\nonumber\\
    & \leq M_{\tilde{n}}(\lambda)\sum \limits_{i=0}^{t-1-h(t)} \Big[ M_{\tilde{n}}(\lambda)(1-G(d_t))\Big]^{i}+\exp{\Big(\frac{\lambda^2\sigma^2 h(t)}{2}\Big)} \Big[M_{\tilde{n}} (\lambda)(1-G(d_t))\Big]^{t-h(t)}\ \quad\text{(from Lemma~\ref{ordering_sg})} \label{c_rec_sg}
\end{align}
\end{comment}
which upon expansion gives an upper bound for $\mathbb{E}_0[e^{\lambda Y(t)}|A_{t-1}]$ as follows.
\begin{align*}
    \mathbb{E}_0[e^{\lambda Y(t)}&|A_{t-1}]\nonumber\leq M_{\tilde{n}}(\lambda)\\&+\sum \limits_{i=h(t)}^{t-2} M_{\tilde{n}}(\lambda)^{t-i} \prod_{j=0}^{t-2-i}(1-G(d_{t-1-j}))\nonumber\\&+\mathbb{E}_0[\exp{(\lambda Y(h(t)))}|A_{h(t)}] \prod_{i=h(t)}^{t-1}M_{\tilde{n}} (\lambda)(1-G(d_i)). 
\end{align*}
Thus, from Claim~\ref{ordering_sg}, we have
\begin{align}
    \mathbb{E}_0[e^{\lambda Y(t)}&|A_{t-1}]\nonumber\leq M_{\tilde{n}}(\lambda)\\&+\sum \limits_{i=h(t)}^{t-2} M_{\tilde{n}}(\lambda)^{t-i} \prod_{j=0}^{t-2-i}(1-G(d_{t-1-j}))\nonumber\\
    &+\mathbb{E}_0[\exp{(\lambda Y'(h(t)))}] \prod_{i=h(t)}^{t-1}M_{\tilde{n}} (\lambda)(1-G(d_i)).\label{rec_sg_main} 
\end{align}
Simplifying the second term in (\ref{rec_sg_main}),
\begin{align*}
    \sum \limits_{i=h(t)}^{t-2} M_{\tilde{n}}(\lambda)^{t-i} &\prod_{j=0}^{t-2-i}(1-G(d_{t-1-j}))\\&\leq
    \sum \limits_{i=h(t)}^{t-2} M_{\tilde{n}}(\lambda)^{t-i} \prod_{j=0}^{t-2-i}(1-G(d_{t}))\\
    &= \sum \limits_{i=h(t)}^{t-2} M_{\tilde{n}}(\lambda)^{t-i}(1-G(d_t))^{t-i-1}\\
    &= M_{\tilde{n}}(\lambda)\sum \limits_{i=h(t)}^{t-2} \Big[M_{\tilde{n}}(\lambda)(1-G(d_t))\Big]^{t-i-1}\\
    &=M_{\tilde{n}}(\lambda)\sum \limits_{i=1}^{t-1-h(t)} \Big[ M_{\tilde{n}}(\lambda)(1-G(d_t))\Big]^{i}.
\end{align*}
Simplifying the product in the third term of (\ref{rec_sg_main}),
\begin{align*}
    \prod_{i=h(t)}^{t-1}M_{\tilde{n}} (\lambda)(1-G(d_i))&\leq \prod_{i=h(t)}^{t-1}M_{\tilde{n}} (\lambda)(1-G(d_t))\\
    &=\Big[M_{\tilde{n}} (\lambda)(1-G(d_t))\Big]^{t-h(t)}.
\end{align*}
Using Claim~\ref{ordering_sg}, the third term in (\ref{rec_sg_main}) is now upper-bounded by $\exp{\Big(\frac{\lambda^2\sigma^2 h(t)}{2}\Big)} \Big[M_{\tilde{n}} (\lambda)(1-G(d_t))\Big]^{t-h(t)}$.
Putting all these simplified terms together in (\ref{rec_sg_main}), we get
\begin{align}
   \mathbb{E}_0[e^{\lambda Y(t)}|A_{t-1}]& \leq M_{\tilde{n}}(\lambda)\sum \limits_{i=0}^{t-1-h(t)} \Big[ M_{\tilde{n}}(\lambda)(1-G(d_t))\Big]^{i}\nonumber\\&+\exp{\Big(\frac{\lambda^2\sigma^2 h(t)}{2}\Big)} \Big[M_{\tilde{n}} (\lambda)(1-G(d_t))\Big]^{t-h(t)}.\label{c_rec_sg}
\end{align}
Since $\tilde{n}(t) \in \mathcal{SG}(\sigma^2)$,
$M_{\tilde{n}}(\lambda) \leq \exp{\Big(\frac{\lambda^2\sigma^2}{2}\Big)}$. We know that $\exp{(x)}\le \frac{1}{1-x}$ for all $x\ge0$. This means that $\exp{\Big(\frac{\lambda^2\sigma^2}{2}\Big)}\le \frac{1}{1-\frac{\lambda^2 \sigma^2}{2}}$. Define $\bar{\gamma}(\lambda)=\frac{1}{1-\frac{\lambda^2 \sigma^2}{2}}$.
Now, (\ref{c_rec_sg}) gives
\begin{align}
    \mathbb{E}_0[&e^{\lambda Y(t)}|A_{t-1}]\leq \bar{\gamma}(\lambda)\sum \limits_{i=0}^{t-1-h(t)} \Big[ \bar{\gamma}(\lambda)(1-G(d_t))\Big]^{i}\nonumber\\&+\exp{\Big(\frac{\lambda^2\sigma^2 h(t)}{2}\Big)} \Big[\bar{\gamma} (\lambda)(1-G(d_t))\Big]^{t-h(t)}.\label{rec_e}
\end{align}
Note that $Y(t)$ is symmetric about zero given $A_{t-1}$. By using Chernoff bound, 
\begin{align}
    \mathbb{P}_0(\lvert Y(t) \rvert \geq k|A_{t-1}) &= 2\mathbb{P}_0(Y(t) \geq k|A_{t-1})\nonumber\\  &\leq 2\mathbb{E}_0[e^{\lambda  Y(t) }|A_{t-1}]\exp{(-\lambda k)}.\nonumber
\end{align}
Therefore,
\begin{align}
    \mathbb{P}_0(\lvert Y(t) \rvert \geq k|A_{t-1}) &\leq 2\Big(\bar{\gamma}(\lambda)\sum \limits_{i=0}^{t-1-h(t)} \Big[ \bar{\gamma}(\lambda)(1-G(d_t))\Big]^{i}\nonumber\\+\exp{\Big(\frac{\lambda^2\sigma^2 h(t)}{2}\Big)} &\Big[\bar{\gamma} (\lambda)(1-G(d_t))\Big]^{t-h(t)}\Big) \exp{(-\lambda k)}. \label{c_bound_p_sg}
\end{align}
We choose an optimal $\lambda$ for a better bound in (\ref{c_bound_p_sg}) by setting $\bar{\gamma}( \lambda) (1-G(d_t)) \le 1$. This implies that $\frac{\lambda^2 \sigma^2}{2}\le G(d_t)$ yielding $\lambda \le \frac{\sqrt{2 G(d_t)}}{\sigma}$. We assume $\lambda = \frac{\sqrt{2 G(d_t)}}{\sigma}$. With this choice of $\lambda$, the RHS of (\ref{c_bound_p_sg}) is upper-bounded by
\begin{align*}
     2\Big[\frac{t-h(t)}{1-G(d_t)}+\exp{\Big(G(d_t)h(t)\Big)}\Big] \exp{\Big(- \frac{\sqrt{2 G(d_t)}}{\sigma}k\Big)}. 
\end{align*}
\end{proof}

\section{Proof of Lemma~\ref{complement_sg}}\label{section:complement_sg}
\begin{proof}[\unskip\nopunct]
We have $A_t=\bigcap\limits_{\tau=h(t)}^t \{\lvert Y(\tau) \rvert \leq d_\tau\}$ where $h(t)=t^\zeta, 0<\zeta<1$. Also, $d_\tau=D~\tau^{\frac{1}{2}+\beta'}$ for some $\beta'>0$ and $D> 0$.
\begin{align*}
    \mathbb{P}_0(A_{t-1}^\mathrm{C})&=\mathbb{P}_0\Big[\Big(\bigcap\limits_{\tau=h(t)}^{t-1} \{\lvert Y(\tau) \rvert \leq d_\tau\}\Big)^\mathrm{C}\Big]\\
    &\labelrel={demorgan}\mathbb{P}_0\Big[\bigcup\limits_{\tau=h(t)}^{t-1} \{\lvert Y(\tau) \rvert > d_\tau\}\Big]\\
    &\labelrel\leq{unionbound} \sum_{\tau=h(t)}^{t-1}\mathbb{P}_0(\lvert Y(\tau) \rvert > d_\tau)\\
    &\labelrel\leq{ordering}\sum_{\tau=h(t)}^{t-1}\mathbb{P}_0(\lvert Y'(\tau) \rvert > D~\tau^{\frac{1}{2}+\beta'})\\
    &\labelrel\leq{subGofY}\sum_{\tau=h(t)}^{t-1} 2\exp{(-c'\tau^{2\beta'})}\; \text{for some $c'>0$}\\
    &\leq\sum_{\tau=h(t)}^{t-1} 2\exp{(-c'h(t)^{2\beta'})}.
\end{align*}
Therefore, we have $\mathbb{P}_0(A_{t-1}^\mathrm{C})\leq 2(t-h(t)) \exp{(-c'h(t)^{2\beta'})}.$
% \begin{align*}
    
% \end{align*}
\eqref{demorgan} and \eqref{unionbound} follow De Morgan's law and Boole's inequality (union bound) respectively. Recalling the characteristics of $\{Y'(t), t\ge 0\}$ as discussed in the proof of Claim~\ref{ordering_sg}, we have $\lvert Y(t) \rvert \le_{st} \lvert Y'(t) \rvert$ and $Y'(t) \in \mathcal{SG}(\sigma^2 t)$.  Hence, the inequalities \eqref{ordering} and \eqref{subGofY}.
\end{proof}

\section{Proof of Theorem~\ref{theorem:extendedstar}}
\label{proof:extendedstar}
\begin{proof}[\unskip\nopunct]
Recall the opinion update models of agent $1$ and its follower $f$. Hence, we obtain that $Y_{f1}(t+1)=$
\begin{align}\label{Y_ia}
    \begin{cases}
    \tilde{n}_{f1}(t)&\text{w.p.}\; (1-G(|Y(t)|))\tilde{G}(|Y_{f1}(t)|)\\
    Y_{f1}(t)+\tilde{n}_{f1}(t)&\text{w.p.}\; (1-G(|Y(t)|))(1-\tilde{G}(|Y_{f1}(t)|))\\
    \frac{Y(t)}{2}+\tilde{n}_{f1}(t)&\text{w.p.}\; G(|Y(t)|)\tilde{G}(|Y_{f1}(t)|)\\
    \frac{Y(t)}{2}+Y_{f1}(t)+\tilde{n}_{f1}(t)&\text{w.p.}\; G(|Y(t)|)(1-\tilde{G}(|Y_{f1}(t)|))
    \end{cases}
\end{align}

where $\tilde{n}(t)=n_1(t)-n_2(t), \tilde{n}_{f1}(t)=n_f(t)-n_1(t)\in[-D,D]$ and $t\ge 0$. Define an event $B_t=\bigcap\limits_{\tau=l(t)}^t \{|Y(\tau)|\le \bar{d}_\tau\}$ where $l(t)=t^\xi$ for $0<\xi<1$ and $\bar{d}_\tau=D~\tau^{\frac{1}{2}-\beta}$ for some $\beta>0$. Using $B_t$, we introduce a high probability bound on $Y(t)$ and adapt it in a similar proof technique used to prove Theorem~\ref{theorem_bounded}. Define $\tilde{k}=c~t^{\frac{1}{2}-\tilde{\beta}}$ for $c, \tilde{\beta}>0$. By law of total probability, 
\begin{align}\label{eq:extendedstar1}
\mathbb{P}_0(|Y_{f1}(t)|\ge \tilde{k})&\le \mathbb{P}_0(|Y_{f1}(t)|\ge \tilde{k}|B_{t-1})\mathbb{P}_0(B_{t-1})\nonumber\\
&+\mathbb{P}_0(|Y_{f1}(t)|\ge \tilde{k}|B_{t-1}^C)\mathbb{P}_0(B_{t-1}^C)\nonumber\\
&\le \mathbb{P}_0(|Y_{f1}(t)|\ge \tilde{k}|B_{t-1})+\mathbb{P}_0(B_{t-1}^C) 
\end{align}
The following lemmas are crucial to the proof of Theorem~\ref{theorem:extendedstar}. Detailed proof of Lemmas~\ref{exstar:complement} and \ref{exstar:condition} can be found in Appendix~\ref{lemma:exstar_comp} and \ref{lemma:exstar_condn} respectively.

\begin{lemma}
Assume $G(x)=\frac{G_0}{1+x^\alpha}$ for some $\alpha>0$ and $0<G_0\le1$. With $\tilde{G}(x)\gtrsim \Big(\frac{1}{x}\Big)^{\frac{2}{3}-\tilde{\delta}}$ for some $\tilde{\delta}>0$, 
\begin{align*}
    \mathbb{P}_0(|Y_{f1}(t)|&\ge \tilde{k}|B_{t-1})\\&\le 2\Big(\frac{t-l(t)}{1-\tilde{G}(\tilde{d}_t)}+\exp{\Big(\frac{5}{8}\lambda^2 D^2 l(t)\Big)}\Big)\exp{(-\lambda \tilde{k})}
\end{align*}
\label{exstar:complement}
where $\lambda=\frac{2\ln 2}{\bar{d}_t}$.
\end{lemma}

\begin{lemma}
\label{exstar:condition}
Recall the event $B_t$. With $G(x)\gtrsim \frac{1}{x^{1-\delta}}$ for some $\delta>0$ and $c_1\le \frac{4-2G_0}{1-G_0}$, 
\begin{align*}
    \mathbb{P}_0(B_{t-1}^C)\le c_1t^2\exp{(-\theta l(t)^\epsilon)}
\end{align*}
for all $t$ such that $\frac{\sqrt{2G(Dt)}}{D}\bar{d}_t\ge \theta t^\epsilon$ where $\epsilon\le \frac{\delta}{2}-\beta$.

\end{lemma}

Using the results from Lemmas~\ref{exstar:complement} and \ref{exstar:condition} into equation \eqref{eq:extendedstar1}, the result in Theorem~\ref{theorem:extendedstar} follows.
\end{proof}

\section{Proof of Lemma~\ref{exstar:complement}}
\label{lemma:exstar_comp}
\begin{proof}[\unskip\nopunct]
Let $M_{\tilde{n}}(\lambda)$ be the moment generating function of $\tilde{n}_{f1}(t)$. Recall that $B_t=\bigcap\limits_{\tau=l(t)}^t \{|Y(\tau)|\le \bar{d}_\tau\}$ where $l(t)=t^\xi$ for $0<\xi<1$, and $\bar{d}_\tau=D~\tau^{\frac{1}{2}-\beta}$ for some $\beta>0$ and $D>0$. We want to obtain an upper bound for the conditional probability $\mathbb{P}_0(|Y_{f1}(t)|\ge \tilde{k}|B_{t-1})$. By using Chernoff bound, for $\lambda>0$,
\begin{align}\label{chernoffstar}
    \mathbb{P}_0(|Y_{f1}(t)|\ge \tilde{k}|B_{t-1}) \le 2\mathbb{E}_0[e^{\lambda Y_{f1}(t)}|B_{t-1}]\exp{(-\lambda \tilde{k})}
\end{align}
We now focus on obtaining an upper bound for $\mathbb{E}_0[e^{\lambda Y_{f1}(t+1)}|B_{t}]$ that will lead us to an upper bound for $\mathbb{E}_0[e^{\lambda Y_{f1}(t)}|B_{t-1}]$ in (\ref{chernoffstar}).
\begin{align}
    \mathbb{E}_0[&e^{\lambda Y_{f1}(t+1)}|B_{t}]\nonumber\\&=M_{\tilde{n}}(\lambda)\mathbb{E}_0[(1-G(|Y(t)|))\tilde{G}(|Y_{f1}(t)|)\nonumber\\&\;+e^{\lambda Y_{f1}(t)}(1-G(|Y(t)|))(1-\tilde{G}(|Y_{f1}(t)|))\nonumber\\&\;+e^{\lambda \frac{Y(t)}{2}} G(|Y(t)|)\tilde{G}(|Y_{f1}(t)|)\nonumber\\
    &\;+e^{\lambda \frac{Y(t)}{2}}e^{\lambda Y_{f1}(t)} G(|Y(t)|)(1-\tilde{G}(|Y_{f1}(t)|))|B_{t}]\nonumber
\end{align}
\begin{align}    
    &\le M_{\tilde{n}}(\lambda)\mathbb{E}_0[(1-G(|Y(t)|))\nonumber\\&\;+e^{\lambda Y_{f1}(t)}(1-G(|Y(t)|))(1-\tilde{G}(|Y_{f1}(t)|))\nonumber\\&\;+e^{\lambda \frac{Y(t)}{2}} G(|Y(t)|)\nonumber\\&\;+e^{\lambda \frac{Y(t)}{2}}e^{\lambda Y_{f1}(t)} G(|Y(t)|)(1-\tilde{G}(|Y_{f1}(t)|))|B_{t}]\nonumber\\
    &\le M_{\tilde{n}}(\lambda)\mathbb{E}_0\Big[\Big(e^{\lambda Y_{f1}(t)} (1-\tilde{G}(|Y_{f1}(t)|))+1\Big)\nonumber\\
    &\quad\Big(1-G(|Y(t)|)+e^{\lambda \frac{Y(t)}{2}}G(|Y(t)|)\Big)|B_t\Big]\label{star:expectation}
\end{align}
In order to upper-bound (\ref{star:expectation}), we use the following Claims~\ref{star:maximum} and \ref{star:convex}.
\begin{claim}
\label{star:maximum}
Given $B_t$, $Y_{f1}(t)\le 2D~t^\frac{3}{2}\; \forall t\ge 0$ for $\xi\le\frac{3}{4}$. 
%{\color{red}This has to be adjusted considering that the bound on $Y(t)$ is true only after $l(t)$. That adjustment has to be incorporated in later derivations.}
\end{claim}
\begin{proof}[Proof of Claim~\ref{star:maximum}]
Recall the event $B_t=\bigcap\limits_{\tau=l(t)}^t \{|Y(\tau)|\le \bar{d}_\tau\}$ where $l(t)=t^\xi$ for $0<\xi<1$.
Given $B_{t}$, $Y(\tau)\le D\sqrt{\tau}$ for $\tau\ge l(t)$ and $Y(\tau)\le D\tau$ for $0\le\tau<l(t)$. Since $Y_{f1}(0)=Y(0)=0$ and $n_{f1}(t)\in[-D,D]\;\forall t\ge 0$, $Y_{f1}(t)\le Dt+\frac{D}{2}\Big(\sum\limits_{i=0}^{l(t)-1}{i}+\sum\limits_{i=l(t)}^{t-1}\sqrt{i}\Big)\; \forall t\ge0$. That is, the maximum drift of $Y_{f1}(t)$ (from (\ref{Y_ia})). By using Cauchy-Schwarz inequality on the third term, we obtain that $Y_{f1}(t)\le Dt+\frac{D}{2}\Big(\frac{l(t)(l(t)-1)}{2}+\frac{(t^3-t^2)^\frac{1}{2}}{\sqrt{2}}-\frac{\Big(l(t)(l(t)-1)\Big)^\frac{1}{2}}{\sqrt{2}}\Big)\le Dt+\frac{D}{2}\Big(\frac{l(t)^2}{2}+\frac{t^\frac{3}{2}}{\sqrt{2}}\Big)\le 2D~t^\frac{3}{2}$ for $\xi\le\frac{3}{4}$.
\end{proof}

\begin{claim}
\label{star:convex}
Recall our assumption that $G(x)=\frac{G_0}{1+x^\alpha}$ for some $\alpha>0$ and $0<G_0\le1$. Given $B_t$, $G(|Y(t)|)(e^{\lambda \frac{Y(t)}{2}}-1)$ is monotonically increasing in $Y(t)$. 
\end{claim}
\begin{proof}[Proof of Claim~\ref{star:convex}]
Given $B_t$, assume that $\lambda\propto \frac{1}{\max\limits_{0\le\tau\le t}Y(\tau)}$. Fix $\lambda=\frac{2\ln{2}}{\bar{d}_t}$. Define $h(x)=(e^{\lambda \frac{x}{2}}-1) G(x)$ for $0\le x\le \bar{d}_t$ for all $t\ge 0$. Since we want to maximize $h(x)$, we restrict the domain to non-negative real numbers. Applying the first derivative on $h(x)$, we obtain that $h(.)$ is monotonically increasing in the domain for $\alpha\le 1$. This implies that $h(x)$ is maximized only at the extreme point $x=\bar{d}_t$.  
\end{proof}
\noindent 
Recall that  $\tilde{d}_t=2D~t^{\frac{3}{2}}$ for $D>0$. We now have
\begin{align*}
    \mathbb{E}_0[e^{\lambda Y_{f1}(t+1)}|B_{t}]&
    \le M_{\tilde{n}}(\lambda)\mathbb{E}_0\Big[\Big(e^{\lambda Y_{f1}(t)} (1-\tilde{G}(\tilde{d}_t))+1\Big)\nonumber\\
    &\quad\Big(1-G(\bar{d}_t)+e^{\lambda \frac{\bar{d}_t}{2}}G(\bar{d}_t)\Big)|B_t\Big]
\end{align*}
Similar to the discussion in the proof of Claim~\ref{restriction1}, by virtue of symmetric constraints and symmetric noise  models, the recursion expression reduces to the following only in terms of $G(\cdot)$, $\tilde{G}(\cdot)$ and $M_{\tilde{n}}(\lambda)$: $\mathbb{E}_0[e^{\lambda Y_{f1}(t+1)}|B_{t}]$
\begin{align*}
    \le M_{\tilde{n}}(\lambda)x_t\Big(\mathbb{E}_0\Big[e^{\lambda Y_{f1}(t)}|B_{t-1}\Big](1-\tilde{G}(\tilde{d}_t))+1\Big)
\end{align*}
which upon expansion leads us to an upper bound for $\mathbb{E}_0[e^{\lambda Y_{f1}(t)}|B_{t-1}]$ as follows.
\begin{align}
    \mathbb{E}_0&[e^{\lambda Y_{f1}(t)}|B_{t-1}]\le M_{\tilde{n}}(\lambda)x_{t-1}\nonumber\\
    &+\sum\limits_{i=l(t)}^{t-2}\Big[\prod\limits_{n=0}^{t-i-1}M_{\tilde{n}}(\lambda)x_{t-n-1}\Big]\prod\limits_{j=0}^{t-2-i}(1-\tilde{G}(\tilde{d}_{t-1-j}))\nonumber
\end{align}
\begin{align}\label{rec:star}
+\Big[\prod\limits_{j=l(t)}^{t-1}M_{\tilde{n}}(\lambda)x_j(1-\tilde{G}(\tilde{d}_j))\Big]\mathbb{E}_0[e^{\lambda Y_{f1}(l(t))}|B_{l(t)}]
\end{align}
where $x_t=1-G(\bar{d}_t)+e^{\lambda \frac{\bar{d}_t}{2}}G(\bar{d}_t)$. The recursive expression (\ref{rec:star}) can further be simplified as follows. Recall the choice of $G(\cdot)$. For the chosen functional form of $G(\cdot)$, $x_t$ is monotonically non-decreasing in $t$. Hence, the second term can be upper-bounded by
$\sum\limits_{i=l(t)}^{t-2}(M_{\tilde{n}}(\lambda)x_{t-1})^{t-i}(1-\tilde{G}(\tilde{d}_{t}))^{t-i-1}$. The simplified expression is
\begin{align*}
    \mathbb{E}_0&[e^{\lambda Y_{f1}(t)}|B_{t-1}]\le M_{\tilde{n}}(\lambda)x_{t-1}\nonumber\\
    &+M_{\tilde{n}(\lambda)}x_{t-1}\sum\limits_{i=l(t)}^{t-2}\Big(M_{\tilde{n}}(\lambda)x_{t-1}(1-\tilde{G}(\tilde{d}_{t}))\Big)^{t-i-1}\nonumber\\
    &+\Big(M_{\tilde{n}}(\lambda)x_{t-1}(1-\tilde{G}(\tilde{d}_t))\Big)^{t-l(t)}\mathbb{E}_0[e^{\lambda Y_{f1}(l(t))}|B_{l(t)}].
\end{align*}

We now use the following claim to upper-bound $\mathbb{E}_0[e^{\lambda Y_{f1}(l(t))}|B_{l(t)}]$ in the third term. The proof technique of Claim~\ref{claim:starorder} follows a similar argument used to prove Claim~\ref{ordering_sg}.

\begin{claim}\label{claim:starorder}
Recall the opinion difference process $\{Y_{f1}(t)\}_{t\ge 0}$ and $l(t)=t^\xi$ for $0<\xi<1$. Then,
\begin{align*}
    \mathbb{E}_0[e^{\lambda Y_{f1}(l(t))}|B_{l(t)}]\le\exp{\Big(\frac{5\lambda^2 D^2 l(t)}{8}\Big)}.
\end{align*}
\end{claim}
\begin{proof}[Proof of Claim~\ref{claim:starorder}]
We use the notion of stochastic ordering to prove Claim~\ref{claim:starorder}. Recall Definition~\ref{s_ordering}.
Let $\{Y''(t)\}_{t\ge0}$ denote the process of opinion difference between the leaders when $G=1$.
Correspondingly, let $\{Y'_{f1}(t)\}_{t\ge0}$ denote the process of opinion difference between agent $1$ and its follower $f$ with $\tilde{G}=0$. That is, at time $t$ for $\tilde{G}=0$ and $G=1$,
\begin{align}\label{yiadash}
    Y''(t+1)=\tilde{n}(t)\;
    \text{and}\; Y'_{f1}(t+1)=Y'_{f1}(t)+\frac{Y''(t)}{2}+\tilde{n}_{f1}(t).
\end{align}
It is clear from (\ref{yiadash}) that $Y_{f1}(t)\le_{st} Y'_{f1}(t)$ for all $t$. As exponential function is non-decreasing in its argument for $\lambda>0$, we have
\begin{align*}
    \mathbb{E}_0[\exp{(\lambda Y_{f1}(t))}]\leq \mathbb{E}_0[\exp{(\lambda Y'_{f1}(t))}]\; \forall t
\end{align*}
from Definition~\ref{s_ordering}. Hence, 
\begin{align*}
    \mathbb{E}_0[\exp{(\lambda Y_{f1}(l(t)))}|B_{l(t)}]\leq \mathbb{E}_0[\exp{(\lambda Y'_{f1}(l(t)))}|B_{l(t)}].
\end{align*}
Since $Y_{f1}(0)=Y(0)=0$, $Y'_{f1}(l(t))=\tilde{n}_{f1}(l(t)-1)+\sum\limits_{j=0}^{l(t)-2}\tilde{n}_{f1}(j)+\frac{\tilde{n}(j)}{2}$. We recall that $\{\tilde{n}(t)\}_{t\ge 0}$ and $\{\tilde{n}_{f1}(t)\}_{t\ge 0}$ are i.i.d. Therefore, we have $\mathbb{E}_0[\exp{(\lambda Y'_{f1}(l(t)))}|B_{l(t)}]$
\begin{align*}
    &\le \mathbb{E}_0\Big[\exp{\Big(\lambda\Big( \sum\limits_{j=0}^{l(t)-1}\tilde{n}_{f1}(j)+\frac{\tilde{n}(j)}{2}\Big)\Big)}\Big]
\end{align*}
\begin{align*}   
    &\le\prod\limits_{j=0}^{l(t)-1}\mathbb{E}_0\Big[\exp{\Big(\lambda\Big( \tilde{n}_{f1}(j)+\frac{\tilde{n}(j)}{2}\Big)\Big)}\Big]\\
    &\le \exp{\Big(\frac{5\lambda^2 D^2 l(t)}{8}\Big)}
\end{align*}
following the properties of independent sub-Gaussian random variables.
\end{proof}
By Hoeffding's Lemma \cite[Sec. 2.3]{lugosi}, 
\begin{align*}
    M_{\tilde{n}}(\lambda)\le \exp{\Big(\frac{\lambda^2D^2}{2}\Big)}\le \frac{1}{1-\frac{\lambda^2 D^2}{2}}.
\end{align*}
Let $\bar{\gamma}(\lambda)=\frac{1}{1-\frac{\lambda^2 D^2}{2}}$. Now, along with Claim~\ref{claim:starorder}, we have
\begin{align}\label{rec:star1}
    \mathbb{E}_0&[e^{\lambda Y_{f1}(t)}|B_{t-1}]\le \bar{\gamma}(\lambda)x_{t}\sum\limits_{i=0}^{t-1-l(t)}\Big(\bar{\gamma}(\lambda)x_{t}(1-\tilde{G}(\tilde{d}_{t}))\Big)^{i}\nonumber\\
    &+\Big(\bar{\gamma}(\lambda)x_{t}(1-\tilde{G}(\tilde{d}_t))\Big)^{t-l(t)} \exp{\Big(\frac{5\lambda^2 D^2 l(t)}{8}\Big)}.
\end{align}
By choosing $\lambda=\frac{2\ln{2}}{\bar{d}_t}$, we obtain that $x_t=1+G(\bar{d}_t)$.
From (\ref{chernoffstar}) and (\ref{rec:star1}), we observe that as $\lambda$ increases, the RHS of (\ref{rec:star1}) increases. However, $\exp(-\lambda \tilde{k})$ in (\ref{chernoffstar}) decreases with $\lambda$. 
One way to obtain a useful bound is such that $\bar{\gamma}(\lambda)(1+G(\bar{d}_t))(1-\tilde{G}(\tilde{d}_t))\le 1$. This yields that $G(\bar{d}_t)\le \tilde{G}(\tilde{d}_t)$ for all $t$ and $\lambda\le\frac{\sqrt{2}\tilde{G}(\tilde{d}_t)}{D}$. From this condition on $G(\cdot)$, we choose $\lambda=\min\limits_t\Big(\frac{2\ln{2}}{\bar{d}_t}, \frac{\sqrt{2}\tilde{G}(\tilde{d}_t)}{D} \Big)=\frac{2\ln{2}}{\bar{d}_t}$ for $\alpha\le 1$. With the above considerations, we obtain that
\begin{align*}
    \mathbb{E}_0[e^{\lambda Y_{f1}(t)}|B_{t-1}]&\le \frac{t-l(t)}{1-\tilde{G}(\tilde{d}_t)}+\exp{\Big(\frac{5}{8}\lambda^2 D^2 l(t)\Big)}.
\end{align*}
The final concentration bound follows the result in Lemma~\ref{exstar:complement}.

\end{proof}

\section{Proof of Lemma~\ref{exstar:condition}}
\label{lemma:exstar_condn}
\begin{proof}[\unskip\nopunct]
The proof technique here uses union bound and  De Morgan's law, and follows similar arguments used while proving Lemma~\ref{complement_sg}.
Recall the event $B_t=\bigcap\limits_{\tau=l(t)}^t \{|Y(\tau)|\le \bar{d}_\tau\}$ where $l(t)=t^\xi$ for $0<\xi<1$ and $\bar{d}_\tau=D~\tau^{\frac{1}{2}-\beta}$ for some $\beta>0$. We have
\begin{align*}
    \mathbb{P}_0(B_{t-1}^C)&=\mathbb{P}_0\Big(\Big[\bigcap\limits_{\tau=l(t)}^{t-1} \{|Y(\tau)|\le \bar{d}_\tau\}\Big]^C\Big)\\
    &\labelrel={Demorgan1}\mathbb{P}_0\Big(\bigcup\limits_{\tau=l(t)}^{t-1} \{|Y(\tau)|> \bar{d}_\tau\}\Big)\\
    &\labelrel\le{union1} \sum_{\tau=l(t)}^{t-1}\mathbb{P}_0(\lvert Y(\tau) \rvert > \bar{d}_\tau)\\
    &\labelrel\le{prevbound1} \sum_{\tau=l(t)}^{t-1} c_1 \tau \exp{(-\theta\tau^{\epsilon})}\\
    &\le \sum_{\tau=l(t)}^{t-1} c_1\tau \exp{(-\theta l(t)^{\epsilon})}\\
    &\le c_1t^2\exp{(-\theta l(t)^\epsilon)}.
\end{align*}

This implies $\mathbb{P}_0(B_{t-1}^C)\le c_1t^2\exp{(-\theta l(t)^\epsilon)}$. \eqref{Demorgan1} and \eqref{union1} follow De Morgan's law and union bound respectively. Inequality~\eqref{prevbound1} follows Proposition~\ref{prop:boundedSimpleStatement} for $G(x)\gtrsim \frac{1}{x^{1-\delta}}$ for some $\delta>0$. 
\end{proof}

\bibliographystyle{IEEEtran}
\bibliography{references}

\end{document}